\documentclass[a4paper,reqno]{amsart}
\usepackage{amsfonts,amsmath,amssymb,amsrefs}
\usepackage{amsthm}
\usepackage{xcolor}
\usepackage{graphicx} 
\usepackage{enumerate}
\usepackage{float}
\usepackage[a4paper, total={6.7in, 9in}]{geometry}
\usepackage{caption}
\usepackage{subcaption}
\graphicspath{ {./Figures/} }
\usepackage{pdflscape}

\newcommand{\hide}[1]{}

\usepackage[hidelinks]{hyperref}

\usepackage{datetime}
\usepackage{comment}
\usepackage{soul}
\usepackage{enumitem}
\usepackage[dvipsnames]{xcolor}
\usepackage{tikz}
\usetikzlibrary{plotmarks}
\usetikzlibrary{decorations.pathreplacing}
\usetikzlibrary{matrix,fit,positioning,arrows.meta}
\numberwithin{equation}{section}

\newtheorem{theorem}{Theorem}[section]

\newtheorem{lemma}[theorem]{Lemma}

\newtheorem{proposition}[theorem]{Proposition}
\numberwithin{figure}{section}

\theoremstyle{remark}
\newtheorem{remark}[theorem]{Remark}

\newcommand\Pone{\textrm{P}_{\textrm{I}} }
\newcommand\Ptwo{\textrm{P}_{\textrm{II}} }
\newcommand\Pthree{\textrm{P}_{\textrm{III}}}
\newcommand\Pfour{\textrm{P}_{\textrm{IV}}}
\newcommand\Pfive{\textrm{P}_{\textrm{V}}}
\newcommand\Psix{\textrm{P}_{\textrm{VI}}}
\newcommand\PJ{\textrm{P}_{\textrm{J}}}

\usepackage[normalem]{ulem}

\newcommand{\orcidauthorA}{0000-0001-7504-4444}
\newcommand{\orcidauthorB}{0000-0001-6859-7788}
\newcommand{\orcidauthorC}{0000-0002-0461-7580}

\title[On difference-differential Lax pairs and integrals of Painlev\'e equations]{On difference-differential Lax pairs and integrals of Painlev\'e equations in finite characteristic}
\date{}
\thanks{This research was supported by the Australian Government through the Office of National Intelligence NISDRG grant \#NI240100145.}
\author[Nalini Joshi]{Nalini Joshi}
\thanks{NJ's ORCID ID is \orcidauthorA.}
\address{School of Mathematics and Statistics F07, The University of Sydney, NSW 2006, Australia}
\author{Tomas Lasic Latimer}
\thanks{TL's ORCID ID is \orcidauthorB.}
\author{Pieter Roffelsen}
\thanks{PR's ORCID ID is \orcidauthorC.}
\email{nalini.joshi@sydney.edu.au}
\email{tomas.lasiclatimer@sydney.edu.au}
\email{pieter.roffelsen@sydney.edu.au}

\begin{document}
\begin{abstract}
We collect rank two difference-differential Lax pairs for classical Painlev\'e equations in the literature and put each in $2\times 2$ matrix form with the coefficient matrix of the spectral equation a degree two matrix polynomial. We describe and apply a general method to obtain integrals of motion in characteristic $p$ from these Lax pairs. For every relevant Painlev\'e equation, this leads to a countable list of integrals of motion, with one entry for each prime $p$. 
\end{abstract}

\maketitle

\tableofcontents

\section{Introduction}
The study of integer or rational solutions of polynomial equations leads naturally to the arithmetic study of elliptic functions, which satisfy cubic equations. The beautiful generalisation of elliptic functions to the Painlev\'e transcendents, consequently leads to the fundamental question: does the arithmetic study of elliptic functions also generalise to an arithmetic study of Painlev\'e transcendents? We investigate this question through the linear problems associated with Painlev\'e transcendents in fields of finite characteristic.

There are six equivalent classes of Painlev\'e equations \cite[\href{https://dlmf.nist.gov/32.2}{(32.2)}]{NIST:DLMF}, which we denote by $P_{\textrm{J}}$, with $J=\textrm{I},\ldots, \textrm{VI}$. The class $\Pthree$ further separates into three types, which we denote by their symmetry groups as $\Pthree^{D_6}$  \eqref{eq:piii}, $\Pthree^{D_7}$  \eqref{eq:piiid7} and $\Pthree^{D_8}$  \eqref{eq:piiid8}. In this paper, we focus on the cases $\textrm{J}=\textrm{VI},\textrm{V},\textrm{IV},\textrm{III}^{D_6},\textrm{III}^{D_7},\textrm{III}^{D_8}$. Each Painlev\'e equation arises as a compatibility condition for an associated pair of linear problems, known as Lax pairs. In most studies, these linear problems are taken to be differential equations. Nevertheless, there are alternatives and we start with an equivalent formulation where one of the linear systems in the Lax pair is a difference equation (see equations \eqref{eq:laxgeneral}). Studying their difference-differential Lax pairs in characteristic $p$, for each Painlev\'e equation we obtain a countable list of rational integrals of motion, so that the entries are indexed by the characteristic $p$.
This is in stark contract with the well-known fact that Painlev\'e transcendents cannot satisfy any rational equation in the function and its derivative over the complex numbers, see \cite{umemura1998painleve,nagloo_21,nagloo_pillay_14} and references therein. The integrals of motion each depend on the dependent and independent variables of the relevant Painlev\'e equation as well as its parameters. We construct them by extending the methodology in \cite{joshiroffelsen} to each corresponding difference-differential Lax pair.

We realise each Painlev\'e equation $\PJ$, $\textrm{J}=\textrm{VI}$ \eqref{eq:pvi}, $\textrm{V}^\text{KNY}$ \eqref{eq:pv}, $\textrm{V}^\text{Na}$ \eqref{eq:pvNa}, $\textrm{IV}$ \eqref{eq:piv}, $\textrm{III}^{D_6}$ \eqref{eq:piii}, $\textrm{III}^{D_7}$ \eqref{eq:piiid7}, $\textrm{III}^{D_8}$ \eqref{eq:piiid8}, over an arbitrary field $\mathbf{k}$, as a differential field $K=\mathbf{k}(f,g,t,\alpha_1,\ldots,\alpha_{m_{\textrm{J}}})$, with derivation $D_t$, where $t$ acts as the independent variable, the $\alpha_j$ act as parameters, i.e. $D_t\alpha_j=0$ for $1\leq j\leq m_{\textrm{J}}$, and $\{f,g\}$ act as dependent variables of $\PJ$. Here $m_{\textrm{J}}$ denotes the number of essential parameters of $\PJ$,
\begin{equation*}
    m_\textrm{VI}=4,\quad m_\textrm{V}^\text{KNY}=m_\textrm{V}^\text{Na}=3,\quad m_\textrm{IV}=2,\quad m_\textrm{III}^{D_6}=2,\quad m_\textrm{III}^{D_7}=1,\quad m_\textrm{III}^{D_8}=0.
\end{equation*}
We realise $\Pfive$ in two distinct but isomorphic ways, corresponding to the two inequivalent Lax pairs available for $\Pfive$ in the literature,
$\Pfive^\text{KNY}$ corresponds to the Lax pair found by Kajiwara et al. \cite{KNY2017} and 
$\Pfive^\text{Na}$ corresponds to the Lax pair found by Nakazono \cite{nakazono_lax}.

The Painlev\'e equation $\PJ$
arises as the compatibility condition
\begin{equation*}
    D_t A_\textrm{J}(z)=B_\textrm{J}(z+1)A_\textrm{J}(z)-A_\textrm{J}(z)B_\textrm{J}(z),
\end{equation*}
of a difference-differential Lax pair of the form 
\begin{subequations}\label{eq:laxgeneral}
\begin{align}
Y(z+1)&=A_\textrm{J}(z)Y(z), & A_\textrm{J}(z)&=z^2A_2+zA_1+A_0,\label{eq:laxgeneralA}\\
    D_t Y(z)&=B_\textrm{J}(z)Y(z),  & B_\textrm{J}(z)&=zB_1+B_0,\label{eq:laxgeneralB}
\end{align}
\end{subequations}
where $A_\textrm{J}(z)$ and $B_\textrm{J}(z)$ are $2\times 2$ matrix polynomials in $z$ of degrees $2$ and $1$ respectively.  The Lax pairs are explicitly recalled (from \cites{adler94,nakazono_lax,KNY2017,ormerodrains2017}) in Sections \ref{sect:p6}--\ref{sect:p3d8}, (with $\Psix$ in \S\ref{subsect:linP6}, $\Pfive^\text{KNY}$  in \S\ref{subsect:linP5}, $\Pfive^\text{Na}$  in \S\ref{subsect:linP5Na}, $\Pfour$ in \S\ref{subsect:linP4}, $\Pthree^{D_6}$ in \S\ref{subsect:linP3d6}, $\Pthree^{D_7}$ in \S\ref{subsect:linP3d7}, and $\Pthree^{D_8}$ in \S\ref{subsect:linP3d8}). 
For each pair, the spectral equation \eqref{eq:laxgeneralA} lies in a particular class of difference equations, characterised by conditions on the trace and determinant of $A(z)$, as detailed in Figure \ref{fig:pain degen} and Table \ref{table:spectraldata}. 

\begin{figure}[h]
    \centering
    \begin{tikzpicture}[
    node distance=8mm and 16mm,
    box/.style={draw, rectangle, minimum width=2cm, minimum height=1cm, align=center},
    ->, >=Stealth
]

\node (Z)[align=left]  {$\Pfive^\text{Na}$\\
deg($|A(z)|$)=4
};
\node[left=of Z, xshift=1.5cm] (Am1) {$A_2 = \begin{bmatrix}
    1 & 0\\t & 1
\end{bmatrix}$};

\node[below=of Z] (A)[align=left]  {$\Psix$\\
deg($|A(z)|$)=4
};
\node[left=of A, xshift=1.5cm] (Am1) {$A_2 = \begin{bmatrix}
    1 & 0\\0 & t
\end{bmatrix}$};

\node[below=of A] (E) [align=left]  {$\Pfive^\text{KNY}$\\
deg($|A(z)|$)=3
};
\node[right=of E] (F) [align=left] {$\Pthree^{D_6}$\\
deg($|A(z)|$)=2};
\node[right=of F] (G) [align=left] {$\Pthree^{D_7}$\\
deg($|A(z)|$)=1};
\node[right=of G] (H) [align=left] {$\Pthree^{D_8}$\\
deg($|A(z)|$)=0};
\node[left=of E, xshift=1.5cm] (Em1) {$A_2 = \begin{bmatrix}
    1 & 0\\0 & 0
\end{bmatrix}$};

\node[below=of E] (I)  [align=left] {$\Pfour$\\
deg($|A(z)|$)=3};
\node[left=of I, xshift=1.5cm] (Im1) {$A_2 = \begin{bmatrix}
    0 & 0\\1 & 0
\end{bmatrix}$};


\draw[->] (A.north) -- (Z.south);
\draw[->] (A.south) -- (E.north);
\draw[->] (E.south) -- (I.north);

\draw[->] (E) -- (F);
\draw[->] (F) -- (G);
\draw[->] (G) -- (H);

\end{tikzpicture}
    
    \caption{An illustration of the degeneration scheme of difference-differential Lax pairs for the Painlev\'e equations. Each row has the same matrix $A_2$, which denotes the leading order coefficient of $A(z)=A_0+zA_1+z^2A_2$, and the degree of the determinant of $A(z)$ decreases from left to right.}
    \label{fig:pain degen}
\end{figure}
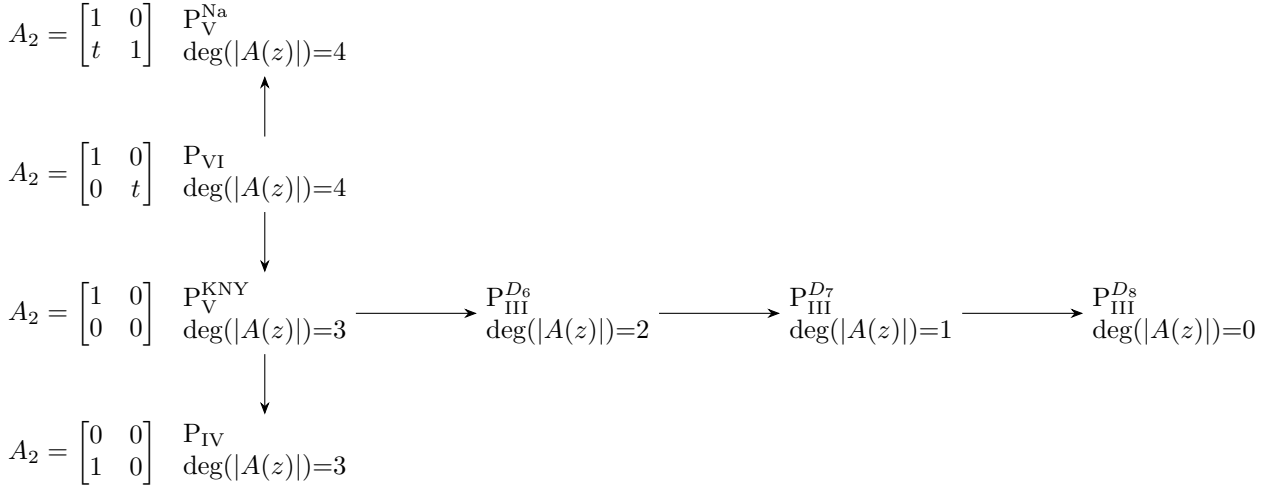

Various difference-differential Lax pairs of rank $2$ are scattered in the literature \cites{adler94,nakazono_lax,KNY2017,ormerodrains2017}, some scalar, others in system form, with different normalisations.
In this paper, we collect and fit all the known rank 2 difference-differential Lax pairs, excluding the symmetric ones in \cite{ormerodrains2017},
for Painlev\'e equations into a unifying framework. In each case the coefficient matrix $A(z)$ of the spectral equation is a degree $2$ matrix polynomial uniquely characterised by the nature of the singularity at infinity and its determinant.
This framework is depicted as a degeneration scheme in Figure \ref{fig:pain degen}.
In this regard, we note that Adler \cite{adler94} constructed rank $2$ difference-differential Lax pairs for $\Psix$, $\Pfour$, $\Pthree^{D_6}$ and $\Pthree^{D_8}$, as well as a rank $3$ one for $\Ptwo$.
Kajiwara et al. \cite{KNY2017} derived scalar rank $2$ difference-differential Lax pairs for $\Pfive$ and $\Pthree^{D_7}$, as well as scalar versions of those found by Adler \cite{adler94}.
Nakazono \cite{nakazono_lax} derived rank $2$ difference-differential Lax pairs for $\Pfour$ and $\Pfive$. The one for $\Pfour$ is equivalent to a coresponding scalar Lax pair given in Kajiwara et al. \cite{KNY2017}, whereas the one derived for $\Pfive$ is not equivalent to the corresponding scalar one in Kajiwara et al. \cite{KNY2017}.
We note that the Lax pairs for $\Psix$ appearing in \cite{ormerodrains2017,arinkinbor2006} are equivalent to the one found earlier by Adler \cite{adler94}. 
Whether $\Ptwo$ admits a rank $2$ difference-differential Lax pair, and whether $\Pone$ admits one at all, seem to be open questions. The classes of linear difference equations, listed in Figure \ref{fig:pain degen} and Table \ref{table:spectraldata}, exhaust all non-rigid rank $2$ linear difference equations of degree $2$ up to scaling and gauge-equivalence.  It follows that $\Ptwo$ and $\Pone$ do not admit a difference-differential Lax pair of rank and degree $2$.

As mentioned, we use these difference-differential Lax pairs from the literature to determine a corresponding integral of motion for each Painlev\'e equation in finite characteristic. The Painlev\'e equations remain relatively unexplored in finite characteristic. They have mostly been studied over the complex numbers as they have their origins in the analytic theory of elliptic functions, special functions and differential equations popularly studied in the nineteenth century.
However, there are some notable results concerning Painlev\'e equations over other fields in the literature. Kanki \textit{et al.} \cite{kankisigma,kankirims} constructed initial value spaces of discrete Painlev\' e equations over $p$-adic fields and finite fields and their relation through `almost good reductions'. Buium and Manin \cite{buium_manin_15} defined several arithmetic analogs of $\Psix$ with `$p$–adic time'. In \cite{joshiroffelsen1} orbits of $q$-difference Painlev\'e one in finite fields are studied and are shown to lie on certain algebraic curves, conjectured to generally be elliptic curves. In \cite{joshiroffelsen}, simultaneous integrals of motion of $\Pfour$ and difference Painlev\'e two are constructed in finite characteristic and geometric properties of these integrals are related to aspects of solutions of the differential and difference equations. In this paper we extend some of the methods in \cite{joshiroffelsen}, providing a  framework to study Painlev\'e equations in finite characteristic.

\renewcommand{\arraystretch}{2.5}
\begin{table}[t]
\centering
\begin{tabular}{|c || c | c | c |} 
 \hline
 P-eqn & Trace of $A(z)$ & Determinant of $A(z)$ & Eigenvalues of $A(z)$ \\
 \hline
$\Psix$ & $\begin{aligned}
    &(1+t)z^2-(\alpha_0+\alpha_4)z\\
    &-t(\alpha_3+\alpha_4-1)z+\mathcal{O}(1)
\end{aligned}$ & $t\,z(z+\alpha_2)(z+\alpha_1+\alpha_2)(z-\alpha_4)$ & $\begin{aligned}
    \lambda_1&=z^2-(\alpha_0+\alpha_4)z+\mathcal{O}(1)\\
    \lambda_2&=t\,z^2-t(\alpha_3+\alpha_4-1)z+\mathcal{O}(1)
\end{aligned}$\\ \hline
$\Pfive^{\text{KNY}}$ & $z^2-(\alpha_1+\alpha_3+t)z+\mathcal{O}(1)$ & $-t\, z(z-\alpha_1)(z+\alpha_2)$ & $\begin{aligned}
    \lambda_1&=z^2-(\alpha_1+\alpha_3)z+\mathcal{O}(1)\\
    \lambda_2&=-t\,z-t(\alpha_2+\alpha_3)+\mathcal{O}(z^{-1})
\end{aligned}$\\ \hline
$\Pfive^{\text{Na}}$ & $\begin{aligned}
    &2z^2+(t+2)z\\
    &-(\alpha_1+2\alpha_2+3\alpha_3)z+\mathcal{O}(1)
\end{aligned}$ & $\begin{aligned} z&(z+1-\alpha_2-\alpha_3)\cdot\\
(z-\alpha_3)&(z+1-\alpha_1-\alpha_2-\alpha_3)
\end{aligned}
$ & $\begin{aligned}
    \lambda_1&=z^2+\sqrt{t}z^{\frac{3}{2}}+\mathcal{O}(z)\\
    \lambda_2&=z^2-\sqrt{t}z^{\frac{3}{2}}+\mathcal{O}(z)
\end{aligned}$\\ \hline
$\Pfour$ & $t\,z+\mathcal{O}(1)$ & $- z(z-\alpha_1)(z+\alpha_2)$ & $\begin{aligned}
    \lambda_1&=+z^{\frac{3}{2}}+\tfrac{1}{2}t\,z+\mathcal{O}(z^{\frac{1}{2}})\\
    \lambda_2&=-z^{\frac{3}{2}}+\tfrac{1}{2}t\,z+\mathcal{O}(z^{\frac{1}{2}})
\end{aligned}$\\ \hline
$\Pthree^{D_6}$ & $z^2+(\alpha_1+\alpha_2-1)z+\mathcal{O}(1)$ & $-t\,z(z+\alpha_2)$ & $\begin{aligned}
    \lambda_1&=z^2+(\alpha_1+\alpha_2-1)z+\mathcal{O}(1)\\
    \lambda_2&=-t-t(1-\alpha_1)z^{-1}+\mathcal{O}(z^{-2})
\end{aligned}$\\ \hline
$\Pthree^{D_7}$ & $z^2+(\alpha_1-1)z+\mathcal{O}(1)$ & $t$\,$z$ & $\begin{aligned}
    \lambda_1&=z^2+(\alpha_1-1)z+\mathcal{O}(1)\\
    \lambda_2&=t\,z^{-1}+t(1-\alpha_1)z^{-2}+\mathcal{O}(z^{-3})
\end{aligned}$\\ \hline
$\Pthree^{D_8}$ & $z^2+\mathcal{O}(1)$ & $t$ & 
$\begin{aligned}
    \lambda_1&=z^2+\mathcal{O}(1)\\
    \lambda_2&=t\,z^{-2}+\mathcal{O}(z^{-4})
\end{aligned}$\\ \hline

\end{tabular}
\caption{Spectral data around $z=\infty$ of the coefficient matrices $A(z)$ in the difference-differential Lax pairs \eqref{eq:laxgeneral} of the different Painlev\'e equations.
}
\label{table:spectraldata}
\end{table}
\renewcommand{\arraystretch}{1}

\subsection{Outline:}
The paper is structured as follows.
In Section \ref{sect:main res}, we describe the main results of the paper. Each of Sections \ref{sect:p6}-\ref{sect:p3d8} is devoted to a particular Painle\'ve equation and corresponding difference-differential Lax pair, $\Psix$ in \S\ref{sect:p6}, $\Pfive^\text{KNY}$ in \S\ref{sect:p5}, $\Pfive^\text{Na}$ in \S\ref{sect:p5Na}, $\Pfour$ in \S\ref{sect:p4}, $\Pthree^{D_6}$ in \S\ref{sect:p3d6}, $\Pthree^{D_7}$  in \S \ref{sect:p3d7}, and $\Pthree^{D_8}$  in \S \ref{sect:p3d8}. In these sections, the spectral equation of each difference-differential Lax pair is characterised and its trace is related to the Hamiltonian of the corresponding Painlev\'e equation, and examples of integrals of motion are given. We provide a conclusion in Section \ref{sec:conclusion} which is followed by the appendices.
Appendix \ref{sec:trac_finite} contains the proof of a technical lemma regarding traces of matrix polynomials in finite characteristic. In Appendix \ref{sec:algebraic} an algebraic solution of $\Psix$ in characteristic $7$ is substituted in an integral of motion for $\Psix$ and the result is shown to indeed be constant.

\subsection{Acknowledgements}
This research was funded by the Australian Government through the Office of National Intelligence grant \# NI240100145.

\section{Main result}\label{sect:main res}
We extract the integrals of motion from difference-differential Lax pairs in finite characteristic. It is based on the following elementary but fundamental proposition.
\begin{proposition}\label{prop: charc inv}
Let $K$ be a differential field of characteristic $p$, with derivation $D_t$, and $A(z)$ and $B(z)$ be matrix polynomials in $z$ over $K$, such that the difference-differential Lax pair,
\begin{align*}
Y(z+1)&=A(z)Y(z), \\
    D_t Y(z)&=B(z)Y(z),
\end{align*}
is compatible with the derivation on $K$, i.e.,
\begin{equation}\label{eq:coefdif}
    D_t A(z)=B(z+1)A(z)-A(z)B(z).
\end{equation}
Define the matrix polynomial
\begin{equation}\label{eq:Mp def}
    M_p(z)=A(z+p-1)\cdot A(z+p-2)\cdot\ldots\cdot A(z+1)\cdot A(z),
\end{equation}
then all the coefficients of the characteristic polynomial $|M_p(z)-\lambda I|$, as a polynomial in $z$ and $\lambda$, are integrals of motion, that is, they lie in the kernel of $D_t$.
\end{proposition}
\begin{proof}
Applying equation \eqref{eq:coefdif} to the definition of $M_p(z)$ in equation \eqref{eq:Mp def} gives
\begin{equation}\nonumber
    D_t M_p(z)= B(z)M_p(z)-M_p(z)B(z).
\end{equation}
Next, we apply $D_t$ to the characteristic polynomial of $M_p(z)$ with $D_tz=D_t\lambda=0$ by definition. Combined with Jacobi's formula for the derivative of a determinant and the above identity, this gives
\begin{align*}
    D_t |M_p(z)-\lambda I|&=\operatorname{Tr}\left[\operatorname{adj}(M_p(z)-\lambda I)\cdot D_t (M_p(z)-\lambda I)\right]\\
    &=\operatorname{Tr}\left[\operatorname{adj}(M_p(z)-\lambda I)\cdot (B(z)M_p(z)-M_p(z)B(z))\right]\\
    &=\operatorname{Tr}\left[\operatorname{adj}(M_p(z)-\lambda I)B(z)M_p(z)]-\operatorname{Tr}[\operatorname{adj}(M_p(z)-\lambda I)M_p(z)B(z)\right]
    \\
     &=0,
\end{align*}
where the last equality follows from the cyclic property of the trace combined with the identity $\operatorname{adj}(A) B=B\operatorname{adj}(A)$ for commuting square matrices $A$ and $B$. So the coefficients of the characteristic polynomial are indeed integrals of motion.
\end{proof}

Proposition \ref{prop: charc inv} demonstrates that all the coefficients of the characteristic polynomial of $M_p(z)$, seen as a polynomial in $z$ and $\lambda$, are integrals of motion. We now consider the case where $K$ is the differential field coming from a Painlev\'e equation $\PJ$, with characteristic $p$, and correspondingly
$A(z)=A_\textrm{J}(z)$, $B(z)=B_\textrm{J}(z)$ as in \eqref{eq:laxgeneral}. In particular, $A(z)$ is a $2\times 2$ matrix and the characteristic polynomial of $M_p(z)$ is completely described by its determinant and trace,
\begin{equation}\label{eq:gen character}
   |M_p(z)-\lambda I|= \lambda^2 - \operatorname{Tr}[M_p(z)]\lambda + |M_p(z)|.
\end{equation}
Table \ref{table:spectraldata} gives the determinant $|A(z)|$ for each Painlev\'e equation $\PJ$. In characteristic $p$, it is a polynomial in $z$ over $\mathbb{F}_p[t,\alpha_1,\ldots,\alpha_{m_\textrm{J}}]$. Correspondingly, $|M_p(z)|$ is a polynomial in $Z$ over $\mathbb{F}_p[T,\mathcal{A}_1,\ldots,\mathcal{A}_{m_\textrm{J}}]$, where
\begin{equation}\label{eq:notationps_reform}
 Z=z^p-z,\quad  T=t^p,\quad \mathcal{A}_k=\alpha_k^p-\alpha_k\quad (k\in\mathbb{Z}).
\end{equation}
For example, $|M_p(z)|=-TZ(Z-\mathcal{A}_1)(Z+\mathcal{A}_2)$ for $\textrm{J}=\textrm{IV}$.
This reaffirms that $D_t|M_p(z)|=0$ and shows that the determinant does not provide any non-trivial integrals of motion of the relevant Painlev\'e equation.
On the other hand, as we will see, the trace of $M_p(z)$ does provide a non-trivial integral of motion. This brings us to our first main result, Theorem \ref{thm:gen trace}, which is based on Proposition \ref{prop: charc inv}, the spectral data characteristing each Lax pair in Table \ref{table:spectraldata} as well as
the following lemma proven in Appendix \ref{sec:trac_finite}.

\begin{lemma}\label{lem:trace}
Let $p$ be an odd prime and $M_p(z)$ be defined as in equation \eqref{eq:Mp def} with $A(z)$ a matrix polynomial of degree two,
\[ A(z) =A_0+z A_1+z^2 A_2 .\]
Then, in characteristic $p$, the trace of $M_p(z)$ takes the form,
    \begin{equation}\label{eq:trace with Ak}
    \operatorname{Tr}M_p(z) =  (\operatorname{Tr}A_2)^pZ^2+
        \operatorname{Tr}[A_1(A_1^{p-1}-A_2^{p-1})]Z+\operatorname{Tr}M_p(0),
\end{equation}
where $ Z=z^p-z$ as in equation \eqref{eq:notationps_reform}.
\end{lemma}

\begin{theorem}\label{thm:gen trace}
Let $p$ be any prime, $\textrm{J}\in\{\textrm{VI},\textrm{V}^\text{KNY},\textrm{V}^\text{Na},\textrm{IV},\textrm{III}^{D_6},\textrm{III}^{D_7},\textrm{III}^{D_8}\}$, and $M_p(z)$ denote the matrix product in \eqref{eq:Mp def} with $A(z)=A_\textrm{J}(z)$. Then the trace of $M_p(z)$ in characteristic $p$ takes the form
\begin{equation}\label{eq:chipform}
   \operatorname{Tr} M_p(z)=\mathcal{I}^{(\textrm{J})}_p+c_1^{(\textrm{J})}Z+c_2^{(\textrm{J})}Z^2,
\end{equation}
where $\mathcal{I}_p^{(\textrm{J})}$ is an integral of motion of the Painlev\'e equation $\PJ$ in characteristic $p$ and the remaining coefficients $c_1^{(\textrm{J})}$ and $c_2^{(\textrm{J})}$ are given by
\begin{equation*}
    c_1^{(\textrm{J})} = \begin{cases} T(\mathcal{A}_3+\mathcal{A}_4)+\mathcal{A}_0+\mathcal{A}_4 & \textrm{J} = \rm{VI}, \\
 T+\mathcal{A}_1+\mathcal{A}_3 & \textrm{J} = \rm{V}^\text{\upshape KNY}, \\
  T-\mathcal{A}_1-2\mathcal{A}_2-3\mathcal{A}_3 & \textrm{J} = \rm{V}^\text{\upshape Na}, \\
 T+\delta_{p,2} & \textrm{J} = \rm{IV}, \\
\mathcal{A}_1+\mathcal{A}_2 & \textrm{J} = \rm{III}^{D_6}, \\
\mathcal{A}_1 & \textrm{J} = \rm{III}^{D_7}, \\
0 & \textrm{J}= \rm{III}^{D_8}, \\
    \end{cases}
    \qquad
        c_2^{(\textrm{J})} = \begin{cases} 1+T & \textrm{J} = \rm{VI}, \\
 1& \textrm{J} = \rm{V}^\text{\upshape KNY}, \\
  2& \textrm{J} = \rm{V}^\text{\upshape Na}, \\
 0 & \textrm{J} = \rm{IV}, \\
1 & \textrm{J} = \rm{III}^{D_6}, \\
1 & \textrm{J} = \rm{III}^{D_7}, \\
1 & \textrm{J}= \rm{III}^{D_8}, \\
    \end{cases}
\end{equation*}
where we used the notation introduced in equation \eqref{eq:notationps_reform} and $\delta_{p,2}=1$ when $p=2$ and zero otherwise.
\end{theorem}
\begin{proof} 
For any odd prime $p$ and choice of $\textrm{J}$,
applying Lemma \ref{lem:trace} and using the data in Table \ref{table:spectraldata}, we obtain that $\operatorname{Tr} M_p(z)$ must take the form in \eqref{eq:chipform}, with coefficients $c_1^{(\textrm{J})}$ and $c_2^{(\textrm{J})}$ as given in the theorem, where we used the Freshman's dream, $(x+y)^p = x^p+y^p$, in characteristic $p$. Applying Proposition \ref{prop: charc inv}, we find that the remaining coefficient $\mathcal{I}_p^{(\textrm{J})}$ is an integral of motion of $\PJ$. The case $p=2$ follows similarly, though in this case the formulas for the coefficients $c_1^{(\textrm{J})}$ and $c_2^{(\textrm{J})}$ are obtained by direct calculation. 
\end{proof}

\begin{remark}\label{rem:invariance_translations}
    We note that the coefficients of \eqref{eq:chipform} are also invariant under the translational elements of the symmetry group of each corresponding Painlev\'e equation. This was worked out in detail in \cite{joshiroffelsen} for $\Pfour$ and the same argument applies here more generally, as each translation can be realised as a Schlessinger transformation of the spectral equation. For any choice of $\textrm{J}$, the invariance under translations can be observed directly for the coefficients $c_1^{(\textrm{J})}$ and $c_2^{(\textrm{J})}$, as they are polynomials in $T$ and $\mathcal{A}_k$, $1\leq k\leq m_\textrm{J}$, and thus invariant under $a_k\mapsto a_k+1$, $1\leq k\leq m_\textrm{J}$.
\end{remark}

\begin{remark}
We note that each of the differential Painlev\'e equations,  normalised as in this paper, reduces to a Riccati equation in characteristic $2$, see e.g. Remark \ref{eq:pvichar2} for the case of $\Psix$. Whether they can be normalised  so as not to become reducible modulo $2$ is an interesting problem that we do not address in this paper.
\end{remark}

From Theorem \ref{thm:gen trace} we deduce that the only possible integral of interest arising from $M_p(z)$ is $\mathcal{I}_p^{(J)}$. We find that $\mathcal{I}_p^{(J)}$ is a non-trivial integral of motion for $\Psix$ -- $\Pthree^{D_8}$. Illustrative examples of the integrals of motion for each Painlev\'e equation are given in its corresponding section.

\section{The Sixth Painlev\'e equation}\label{sect:p6}

\subsection{The nonlinear system}
Fix a field $\mathbf{k}$ and consider the field of rational functions $\mathbf{k}(f,g,\alpha_1,\alpha_2,\alpha_3,\alpha_4,t)$ with
the $\mathbf{k}$-linear derivation $D_t$ specified by
 \begin{equation}\label{eq:pvi}
 \Psix:\quad           \begin{cases}
     D_t \hspace{0.5mm}f\hspace{0.5mm}=\,\displaystyle\frac{2f(f-t)(f-1)g-(\alpha_0-1)f(f-1)-\alpha_3 f(f-t)-\alpha_4(f-t)(f-1)}{t(t-1)},  &\\
      D_t \hspace{0.7mm}g\hspace{0.7mm}=\,\displaystyle\frac{g\big[(\alpha_0+\alpha_3-1)f+(\alpha_3+\alpha_4)(f-t)+(\alpha_0+\alpha_4-1)(f-1)+g\,h\big]-\alpha_2(\alpha_1+\alpha_2)}{t(t-1)}, & \\
  D_t \alpha_j \hspace{-0.5mm}=\,0\quad (1\leq j\leq 4), &\\
        D_t \hspace{1.0mm}t\hspace{1.0mm}=\,1, &   
    \end{cases}
    \end{equation}
where $\alpha_0$ and $h$ are defined through the equations
\begin{equation*}
    \alpha_0+\alpha_1+2\alpha_2+\alpha_3+\alpha_4=1,\quad h=1+t(t-1)-f^2-(f-t)^2-(f-1)^2.
\end{equation*} 
The system of ODEs \eqref{eq:pvi} is known as the sixth Painlev\'e equation (in system form). If the characteristic of $\mathbf{k}$ is not $2$, eliminating $g$ yields the standard scalar version of the sixth Painlev\'e equation for $f$,
\begin{equation} \label{eq:PVIscalar}
\begin{aligned}
D_t^2f=&\frac{1}{2}\left(\frac{1}{f}+\frac{1}{f-1}+\frac{1}{f-t}\right)(D_t f)^2-\left(\frac{1}{t}+\frac{1}{t-1}+\frac{1}{f-t}\right)D_t f\\
&+\frac{f(f-1)(f-t)}{t^2(t-1)^2}\left(\alpha+\frac{\beta\, t}{f^2}+\frac{\gamma\,(t-1)}{(f-1)^2}+\frac{\delta \,t(t-1)}{(f-t)^2}\right),
\end{aligned}
\end{equation}
with
\begin{equation*}
    \alpha=\frac{1}{2}\alpha_1^2,\quad \beta=-\frac{1}{2}\alpha_4^2,\quad
    \gamma=\frac{1}{2}\alpha_3^2,\quad \delta=\frac{1}{2}(1-\alpha_0^2).
\end{equation*}
\begin{remark}\label{eq:pvichar2}
    The case of characteristic $2$ is exceptional since the sixth Painlev\'e equation, scaled as in \eqref{eq:pvi},
    degenerates to a Riccatti equation in characteristic $2$. This follows immediately from the equation for $D_tf$ in equation \eqref{eq:pvi}, which reduces to 
\begin{equation*}
   t(t+1) D_t \hspace{0.6mm}f\hspace{0.86mm}=(\alpha_0+1)f(f+1)+\alpha_3 f(f+t)+\alpha_4(f+t)(f+1),
\end{equation*}
modulo $2$. Similarly,  by multiplying equation \eqref{eq:PVIscalar} by $2$, expanding, and reducing the ODE modulo $2$, it reduces to the above Riccatti equation. Whether $\Psix$ can be rescaled such that it does not degenerate modulo $2$ is an interesting question that we will not address here. Analogous statements apply to each of the other Painlev\'e equations, $\textrm{P}_{\textrm{J}}$, $\textrm{J}=\textrm{V},\textrm{IV},\textrm{III}^{D_6},\textrm{III}^{D_7},\textrm{III}^{D_8}$, as well.
\end{remark}

\subsection{The Lax pair}\label{subsect:linP6}
We consider the following difference-differential Lax pair, taken from \cite{ormerodrains2017}, first derived in \cite{adler94} and also derived in \cite{borodin2006},
\begin{subequations}\label{eq:laxpVI}
\begin{align}
    Y(z+1)=A(z)Y(z),\label{eq:laxpviA}\\
    D_t Y(z)=B(z)Y(z),\label{eq:laxpviB}
\end{align}
\end{subequations}
where the coefficient matrix $A(z)$ is a $2\times 2$ matrix polynomial of degree two in $z$,
\begin{equation}\label{eq:pviAseries}
        A(z)=A_0+z A_1+z^2A_2,
    \end{equation}
    characterised by the following three properties.
\begin{enumerate}[label=(\roman*)]
\item The leading order coefficient is given by
\begin{equation*}
    \qquad A_2=\begin{bmatrix}
            1 & 0\\
            0 & t
        \end{bmatrix}.
\end{equation*}
    \item The trace of $A(z)$ has leading order behaviour
   \begin{equation}\label{eq:traceinf}
    \begin{aligned}
     \operatorname{Tr}A(z)=(1+t)z^2-(\alpha_0+\alpha_4+t(\alpha_3+\alpha_4-1))z+\mathcal{O}(1),
    \end{aligned}
   \end{equation}
    as $z\rightarrow \infty$.
    \item The determinant of $A(z)$ equals
\begin{equation}\label{eq:coefmatrixdet}
    |A(z)|=t\,z(z+\alpha_2)(z+\alpha_1+\alpha_2)(z-\alpha_4).
\end{equation}
\end{enumerate}
In other words, the parameters $\{\alpha_1,\alpha_2,\alpha_4\}$ can be thought of as controlling the locations of the zeros of the determinant of $A(z)$, whereas $\{t,\alpha_3\}$ characterise the singularity at $z=\infty$ of equation \eqref{eq:laxpviA}. We use standard coordinates $\{f,g,w\}$ on the coefficient matrix, defined by 
\begin{equation*}
    A_{12}(z)=w(z-f\,g),\qquad A_{11}(z)|_{z=fg}\,=t\,g\,(f\,g-\alpha_4),
\end{equation*}
where $A_{12}(z)$ denotes the $(1,2)$-entry of $A(z)$ and $A_{11}(z)|_{z=fg}$ denotes the $(1,1)$-entry of $A(z)$ evaluated at $z=fg$. These coordinates combined with properties (i)--(iii) result in $A_1$ and $A_0$ in \eqref{eq:pviAseries} to be given by
\begin{align*}
A_1&=\begin{bmatrix}
 -(\alpha _0+\alpha _4) & w \\
\gamma_1 \,w^{-1} & -t(\alpha _3+\alpha _4-1)
\end{bmatrix},\\
A_0&=\begin{bmatrix}
 g [\alpha _0 f+\alpha _4 (f-t)-f g (f-t)] & -f g\, w \\
 \gamma_0 \,w^{-1} & \delta\, f
\end{bmatrix},
\end{align*}
where
\begin{align*}
    \gamma_1&=\alpha _0 (\alpha _3-1) t+\alpha _2 (\alpha _1+\alpha _2) (f-t)+(f-t)g\big[f (f-t)g-(\alpha _0+\alpha_3-1) f-\alpha _4(f-t)\big],\\
    \gamma_0&=\big(\alpha _0 f-(f-t) (f g-\alpha _4)\big)  \cdot\big(\alpha _0 f g+(\alpha _3+\alpha _4-1)  (f-t)g-f  (f-t)g^2-\alpha _2 (\alpha _1+\alpha _2)\big),\\
    \delta&=f  (f-t)g^2-\alpha _0 f g-(\alpha _3+\alpha _4-1)  (f-t)g+\alpha _2 (\alpha _1+\alpha _2).
\end{align*}
In particular, $A(z)$ is a matrix polynomial in $z$ with coefficient entries from $\mathbb{Z}[f,g,\alpha_1,\alpha_2,\alpha_3,\alpha_4,t,w,w^{-1}]$. The corresponding degree one matrix polynomial $B(z)$ in \eqref{eq:laxpviB} is given by
\begin{equation*}
    B(z)=\frac{1}{t(t-1)}\begin{bmatrix}
 0 & w \\
 \gamma_1\,w^{-1} & 0 \\
\end{bmatrix}+\frac{z}{t}\begin{bmatrix}
 0 & 0 \\
 0 & 1\\
\end{bmatrix}.
\end{equation*}
 Furthermore, compatibility of the Lax pair \eqref{eq:laxpVI} implies $\Psix$ and that the auxiliary variable $w$ satisfies the following ODE,
\begin{equation*}
    \frac{D_tw}{w}=\frac{1}{t(t-1)}\big(\alpha_0+\alpha_4-t\,(\alpha_3+\alpha_4-1)+(t-1)f\,g\big).
\end{equation*}

\subsection{The Hamiltonian} Whilst the trace of the product of $A(z)$ evaluated at successive integers,
$$\operatorname{Tr}[A(n-1)\cdot A(n-2)\cdot\ldots\cdot A(1)\cdot A(0)],$$
gives an integral of motion in characteristic $p$ when $n=p$ as in Theorem \ref{thm:gen trace}, the case $n=1$ is related to the Hamiltonian of the sixth Painlev\'e equation. Namely,
\begin{align*}
    H:&=\operatorname{Tr}A_0+(t-1)f\,g\\
    &=\alpha_2(\alpha_1+\alpha_2)f-(\alpha_0-1)f(f-1)g-\alpha_3f(f-t)g-\alpha_4(f-t)(f-1)g+f(f-t)(f-1)g^2,
\end{align*}
is the Hamiltonian of $\Psix$, that is,
\begin{align*}
    t(t-1)\,D_tf&=+\frac{\partial H}{\partial g},\\
    t(t-1)\,D_tg&=-\frac{\partial H}{\partial f},
\end{align*}
are respectively equivalent to the equations for $D_tf$ and $D_tg$ in \eqref{eq:pvi}.

\subsection{An integral in characteristic 2}
The first integral of motion $\mathcal{I}_2^{\textrm{VI}}$, for $\Psix$ in Theorem \ref{thm:gen trace}, is given by
\begin{align*}
      \mathcal{I}_2^{\textrm{VI}}=\,&
+f^6 g^4 + (1 + t^2) f^4 g^4 + \alpha_1^2 f^4 g^2 + 
 t^2 f^2 g^4 + (1 + \alpha_1 + \alpha_3 + t(\alpha_3 + \alpha_4) ) f^3 g^2 \\
 &+ (1 + \alpha_1 + \alpha_1^2 + \alpha_3 + 
    \alpha_3^2 + t(\alpha_1 + \alpha_4) + t^2(1 + \alpha_3 + \alpha_3^2 + \alpha_4 + \alpha_4^2) ) f^2 g^2 \\
    &+ 
 \alpha_1 (1 + \alpha_1 + \alpha_3 + t(\alpha_3 + \alpha_4) ) f^2 g 
 + t(\alpha_1 + \alpha_3 +t (1 + \alpha_3 + \alpha_4) )  f g^2 + 
  (\alpha_1 + \alpha_2)^2\alpha_2^2 f^2\\
 &+ (\alpha_1 + \alpha_1^2 + \alpha_3 + 
    \alpha_3^2 + t^2(\alpha_3 + \alpha_3^2 + \alpha_4 + \alpha_4^2) ) f g + 
 \alpha_4^2 t^2 g^2+ (\alpha_1 + \alpha_2) \alpha_2 (1 + \alpha_1 + \alpha_3 + t(\alpha_3 + \alpha_4)) f \\
 &+ 
 \alpha_4 t (\alpha_1 + \alpha_3 + t(1 + \alpha_3 + \alpha_4) ) g + (\alpha_1 + \alpha_2) \alpha_2 \alpha_4 t.
\end{align*}
We may drop all terms polynomial in $f^p,g^p,t^p$ and parameters from an integral of motion in characteristic $p$, since such terms trivially lie in the kernel of $D_t$. This yields e.g. the following simplified integral of motion for $\Psix$ in characteristic $2$,
\begin{align*}
    \widetilde{\mathcal{I}}_2^{\textrm{VI}}=\,&
 + (1 + \alpha_1 + \alpha_3 + t(\alpha_3 + \alpha_4) ) f^3 g^2 + t(\alpha_1 + \alpha_4) f^2 g^2 + 
 \alpha_1 (1 + \alpha_1 + \alpha_3 + t(\alpha_3 + \alpha_4) ) f^2 g \\
 &
 + t(\alpha_1 + \alpha_3 +t (1 + \alpha_3 + \alpha_4) )  f g^2 + (\alpha_1 + \alpha_1^2 + \alpha_3 + 
    \alpha_3^2 + t^2(\alpha_3 + \alpha_3^2 + \alpha_4 + \alpha_4^2) ) f g \\
    &+ (\alpha_1 + \alpha_2) \alpha_2 (1 + \alpha_1 + \alpha_3 + t(\alpha_3 + \alpha_4)) f + 
 \alpha_4 t (\alpha_1 + \alpha_3 + t(1 + \alpha_3 + \alpha_4) ) g + (\alpha_1 + \alpha_2) \alpha_2 \alpha_4 t.
\end{align*}
We note, however, that Remark \ref{rem:invariance_translations} no longer applies to such simplified integrals of motion. That is, they are no longer necessarily invariant under translational elements of the corresponding symmetry group.

In Appendix \ref{sec:algebraic}, the integral $\mathcal{I}_7^{\textrm{VI}}$ is given for a particular choice of the parameters $\alpha_i$, $1\leq i\leq 4$, corresponding to an algebraic solution. It is shown in that section that the integral is indeed constant upon substitution of the relevant algebraic solution.

\section{The Fifth Painlev\'e equation (KNY)}\label{sect:p5}

\subsection{The nonlinear system}
Fix a field $\mathbf{k}$ and consider the field of rational functions $\mathbf{k}(f,g,\alpha_1,\alpha_2,\alpha_3,t)$ with the $\mathbf{k}$-linear derivation $D_t$ specified by
 \begin{equation}\label{eq:pv}
 \Pfive^\text{KNY}:\quad           \begin{cases}
       t\, D_t f=-t\,f+2(f-1)fg-(\alpha_1+\alpha_3)f+tf^2+\alpha_1, &   D_t \alpha_1=0, \\
  t\, D_t \hspace{0.2mm}g\hspace{0.1mm}=+t\,\hspace{0.2mm}g\hspace{0.1mm}-2(\hspace{0.2mm}g\hspace{0.1mm}+t)fg+(\alpha_1+\alpha_3)\hspace{0.2mm}g+g^2-\alpha_2 t, &   D_t \alpha_2=0, \\
      \hspace{1.9mm}    D_t \hspace{0.4mm}t\hspace{0.4mm}=1, &   D_t \alpha_3=0. \\ 
    \end{cases}
    \end{equation}
This system of ODEs is known as the fifth Painlev\'e equation (in system form). Assuming the characteristic of $\mathbf{k}$ is not $2$, eliminating $g$ from equations \eqref{eq:pv}, gives the standard scalar form of the fifth Painlev\'e equation for the transformed dependent variable $y=1-f^{-1}$,
\begin{equation}\label{eq:pvscal}
    D_t^2y=\left(\frac{1}{2y}+\frac{1}{y-1}\right)(D_ty)^2-\frac{1}{t}D_ty+\frac{(y-1)^2}{t^2}\left(\alpha y+\frac{\beta}{y}\right)+\gamma \frac{y}{t}+\delta \frac{y(y+1)}{y-1},
\end{equation}
with
\begin{equation}\label{eq:pvscal_par}
    \alpha=\frac{1}{2}\alpha_1^2,\quad \beta=-\frac{1}{2}\alpha_3^2,\quad
    \gamma=1-(\alpha_1+2\alpha_2+\alpha_3),\quad \delta=-\frac{1}{2}.
\end{equation}

\subsection{The Lax pair} \label{subsect:linP5}
We consider the following difference-differential Lax pair, obtained by rewriting the one given by Kajiwara et al. \cite{KNY2017} for $\Pfive$ in system form,
\begin{subequations}\label{eq:laxpv}
\begin{align}
    Y(z+1)=A(z)Y(z),\label{eq:laxpvA}\\
    D_t Y(z)=B(z)Y(z),\label{eq:laxpvB}
\end{align}
\end{subequations}
where the coefficient matrix $A(z)$ is a $2\times 2$ matrix polynomial of degree two in $z$,
\begin{equation}\label{eq:pvAseries}
        A(z)=A_0+z A_1+z^2A_2,
    \end{equation}
    characterised by the following three properties.
\begin{enumerate}[label=(\roman*)]
\item The leading order coefficient is given by
\begin{equation*}
    \qquad A_2=\begin{bmatrix}
            1 & 0\\
            0 & 0
        \end{bmatrix}.
\end{equation*}
    \item The trace of $A(z)$ has leading order behaviour
   \begin{equation}\label{eq:traceinfpvKNY}
    \begin{aligned}
     \operatorname{Tr}A(z)=z^2-(t+\alpha_1+\alpha_3)z+\mathcal{O}(1),
    \end{aligned}
   \end{equation}
    as $z\rightarrow \infty$.
    \item The determinant of $A(z)$ equals
\begin{equation}\label{eq:coefmatrixdetpvKNY}
    |A(z)|=-t\,z(z-\alpha_1)(z+\alpha_2).
\end{equation}
\end{enumerate}
In other words, the parameters $\{\alpha_1,\alpha_2\}$ can be thought of as controlling the locations of the zeros of the determinant of $A(z)$, whereas $\{t,\alpha_3\}$ characterises the singularity at $z=\infty$ of equation \eqref{eq:laxpvA}. We use standard coordinates $\{f,g,w\}$ on the coefficient matrix, defined by 
\begin{equation*}
    A_{12}(z)=w(z-f\,g),\qquad A_{11}(z)|_{z=fg}\,=g\,(f\,g-\alpha_1),
\end{equation*}
where $A_{12}(z)$ denotes the $(1,2)$-entry of $A(z)$ and $A_{11}(z)|_{z=fg}$ denotes the $(1,1)$-entry of $A(z)$ evaluated at $z=fg$. 
These coordinates combined with properties (i)--(iii) result in $A_1$ and $A_0$ in \eqref{eq:pvAseries} to be given by
\begin{align*}
A_1&=\begin{bmatrix}
 -(\alpha_1+\alpha_3) & w\\
 [\alpha_3-(f-1)(f\,g+\alpha_2)]\,t\,w^{-1} & -t
\end{bmatrix},\\
A_0&=\begin{bmatrix}
g\, h & -f\, g\, w\\
[\alpha_2+(f-1)g]\,h\,t\,w^{-1} & -t\,f[\alpha_2+(f-1)g]
\end{bmatrix},
\end{align*}
where
\begin{equation*}
    h=\alpha_3\, f+(f-1)(\alpha_1-f\,g).
\end{equation*}
In particular, $A(z)$ is a matrix polynomial in $z$ with coefficient entries from $\mathbb{Z}[f,g,\alpha_1,\alpha_2,\alpha_3,t,w,w^{-1}]$. The corresponding degree one matrix polynomial $B(z)$ in \eqref{eq:laxpvB} is given by,
\begin{equation*}
    B(z)=-\frac{1}{t}\begin{bmatrix}
 0 & w \\
 [\alpha_3-(f-1)(\alpha_2+f g)]\,t\,w^{-1} & 0 \\
\end{bmatrix}+\frac{z}{t}\begin{bmatrix}
 0 & 0 \\
 0 & 1\\
\end{bmatrix}.
\end{equation*}
Furthermore, compatibility of the Lax pair \eqref{eq:laxpv} is equivalent to $\Pfive$ and the following ODE for the auxiliary variable $w$,
\begin{equation*}
    \frac{D_tw}{w}=1-\frac{\alpha_1+\alpha_3-f\,g}{t}.
\end{equation*}

\subsection{The Hamiltonian} Whilst the trace of the product of $A(z)$ evaluated at successive integers,
$$\operatorname{Tr}[A(n-1)\cdot A(n-2)\cdot\ldots\cdot A(1)\cdot A(0)],$$
gives an integral of motion in characteristic $p$ when $n=p$ as in Theorem \ref{thm:gen trace}, the case $n=1$ is related to the Hamiltonian of $\Pfive^\text{KNY}$. Namely,
\begin{equation*}
    H:=-\operatorname{Tr}A_0=f(f-1)g(g+t)-\alpha_1(f-1)g+\alpha_2 t f-\alpha_3 f g,
\end{equation*}
is the Hamiltonian of $\Pfive^\text{KNY}$, that is,
\begin{align*}
    t\,D_tf&=+\frac{\partial H}{\partial g},\\
    t\,D_tg&=-\frac{\partial H}{\partial f},
\end{align*}
are respectively equivalent to the equations for $D_tf$ and $D_tg$ in \eqref{eq:pv}.

\subsection{Integrals in finite characteristic}
The first few integrals of motion $\mathcal{I}_p^{\textrm{V}(\text{NKY})}$, $p=2,3$, for $\Pfive^\text{KNY}$ in Theorem \ref{thm:gen trace}, are given by
\begin{align*}
    \mathcal{I}_2^{\textrm{V}(\text{KNY})}=\,&+f^4 g^4+(t f+g)^2f^2 g^2 + (\alpha_1^2+\alpha_1+\alpha_3^2+\alpha_3+t^2+t+1)f^2 g^2+f^2g+(\alpha_1+\alpha_3+t)(tf+g)fg\\
    &+(\alpha_1^2+\alpha_1+\alpha_3^2+\alpha_3+t^2)fg+\alpha_2^2 t^2 f^2+\alpha_1^2 g^2+(\alpha_1+\alpha_3+t)(\alpha_2 t f+\alpha_1 g)+\alpha_1 g+\alpha_1 \alpha_2 t,\\
   \mathcal{I}_3^{\textrm{V}(\text{KNY})}=\, &-f^6g^6- (t f-g)^3f^3 g^3-f^4 g^4+f^3 g^3 (t f-g)-(t^2 f^2+g^2)f^2 g^2-(\alpha_1-\alpha_2+\alpha_3+t-1)t f^3 g^2\\
   &+(\alpha_1^3-\alpha_1+\alpha_3^3-\alpha_3+t^3+t)f^3 g^3+(\alpha_3-\alpha_1+t) f^2 g^3-(\alpha_1 g^2+f g-\alpha_2 t^2 f^2)f g\\
   &+(\alpha_1(\alpha_1-\alpha_2-\alpha_3+t+1)+\alpha_2+\alpha_3(\alpha_3-\alpha_2-t+1)+t^2)t f^2 g- \alpha_2^3 t^3 f^3- \alpha_1^3 g^3\\
   & +(\alpha_1^2-\alpha_3^2
   +(\alpha_1-\alpha_2+\alpha_3-t-1)t+1)f g^2-\alpha_2^2 t^2 f^2-\alpha_1^2 g^2\\
   &-(\alpha_1^3-\alpha_1+\alpha_3^3-\alpha_3-(\alpha_1^2+\alpha_1 \alpha_3+ \alpha_1-\alpha_2 \alpha_3-\alpha_2)t+ (\alpha_2-\alpha_1)t^2+t^3)f g \\
   &+(\alpha_1^2- \alpha_1 \alpha_3+\alpha_3^2+\alpha_1+\alpha_3+ (\alpha_1+\alpha_2-\alpha_3)t+t^2)\alpha_2tf\\
   &+(\alpha_1^2-\alpha_1 \alpha_3+\alpha_3^2-1+(\alpha_1+\alpha_2-\alpha_3+1)t+t^2)\alpha_1g-\alpha_1 \alpha_2(\alpha_1+\alpha_3+t+1)t.
\end{align*}

\section{The Fifth Painlev\'e equation (Nakazono)}\label{sect:p5Na}

\subsection{The nonlinear system}
Fix a field $\mathbf{k}$ and consider the field of rational functions $\mathbf{k}(f,g,\alpha_1,\alpha_2,\alpha_3,t)$ with the $\mathbf{k}$-linear derivation $D_t$ specified by
 \begin{equation}\label{eq:pvNa}
 \Pfive^\text{Na}:\quad           \begin{cases}
       t\, D_t f=-t\,f+(f-1)(\alpha_3+(2\alpha_0+\alpha_1)f+2f(f-1)g, &   D_t \alpha_1=0, \\
  t\, D_t \hspace{0.2mm}g\hspace{0.1mm}=+t\,\hspace{0.2mm}g\hspace{0.1mm}-\alpha_0(\alpha_0+\alpha_1)t-g[\alpha_3+(2\alpha_0+\alpha_1)(2f-1)+(f-1)(3f-1)g], &   D_t \alpha_2=0, \\
     \hspace{1.9mm}   D_t \hspace{0.4mm}t\hspace{0.4mm}=1, &   D_t \alpha_3=0, \\ 
    \end{cases}
    \end{equation}
with $\alpha_0$ defined by $\alpha_0+\alpha_1+\alpha_2+\alpha_3=1$.    
This system of ODEs is another system form of the fifth Painlev\'e equation. Assuming the characteristic of $\mathbf{k}$ is not $2$, eliminating $g$ from equations \eqref{eq:pvNa}, gives the standard scalar form \eqref{eq:pvscal} of the fifth Painlev\'e equation  for $y=f$ with parameters \eqref{eq:pvscal_par}.

Temporarily writing $\{f^\text{KNY},g^\text{KNY}\}$ for the depend variables in $\Pfive^\text{KNY}$, the differential fields $\Pfive^\text{Na}$ and $\Pfive^\text{KNY}$ are isomorphic under the change of dependent variables
\begin{equation}\label{eq:NatoKNY}
f=1-1/f^\text{KNY},\quad g=f^\text{KNY}(\alpha_0+f^\text{KNY}(t+g^\text{KNY})).
\end{equation}

\subsection{The Lax pair} \label{subsect:linP5Na}
We consider the following difference-differential Lax pair, obtained by renormalising and gauge transforming the one derived by Nakazono
\cite{nakazono_lax} for $\Pfive$,
\begin{subequations}\label{eq:laxpvNa}
\begin{align}
    Y(z+1)=A(z)Y(z),\label{eq:laxpvNaA}\\
    D_t Y(z)=B(z)Y(z),\label{eq:laxpvNaB}
\end{align}
\end{subequations}
where the coefficient matrix $A(z)$ is a $2\times 2$ matrix polynomial of degree two in $z$,
\begin{equation*}
        A(z)=A_0+z A_1+z^2A_2,
    \end{equation*}
    characterised by the following four properties.
\begin{enumerate}[label=(\roman*)]
\item The leading order coefficient is given by
\begin{equation*}
    \qquad A_2=\begin{bmatrix}
            1 & 0\\
            t & 1
        \end{bmatrix}.
\end{equation*}
    \item The trace of $A(z)$ has leading order behaviour
   \begin{equation}\label{eq:traceinfpvNAK}
    \begin{aligned}
     \operatorname{Tr}A(z)=2\,z^2+(t+2-\alpha_1-2\alpha_2-3\alpha_3)z+\mathcal{O}(1),
    \end{aligned}
   \end{equation}
    as $z\rightarrow \infty$.
    \item The determinant of $A(z)$ equals
\begin{equation}\label{eq:coefmatrixdetpvNAK}
    |A(z)|=z(z-\alpha_3)(z+1-\alpha_2-\alpha_3)(z+1-\alpha_1-\alpha_2-\alpha_3),
\end{equation}
\item The top-left entry of $A_1$ is zero, $(A_1)_{11}=0$. 
\end{enumerate}
In other words, the parameters $\{\alpha_1,\alpha_2,\alpha_3\}$ can be thought of as controlling the locations of the zeros of the determinant of $A(z)$, whereas $t$ characterises the singularity at $z=\infty$ of equation \eqref{eq:laxpvNaA}.
 We note that, given (i) and (iii), property (ii) is equivalent to the top-right entry of $A_1$ being equal to one, $(A_1)_{12}=1$.

We introduce standard coordinates $\{f,g\}$ on the coefficient matrix, defined by 
\begin{equation*}
    A_{12}(z)=(z-f\,g),\qquad A_{11}(z)|_{z=fg}\,=g(f\,g-\alpha_3),
\end{equation*}
where $A_{12}(z)$ denotes the $(1,2)$-entry of $A(z)$ and $A_{11}(z)|_{z=fg}$ denotes the $(1,1)$-entry of $A(z)$ evaluated at $z=fg$. 
These coordinates combined with properties (i)--(iv) result in $A_1$ and $A_0$ in \eqref{eq:piiiAseries} to be given by
\begin{align*}
A_1&=\begin{bmatrix}
    0 & 1\\
    \gamma_1 & t+2\alpha_0+\alpha_1-\alpha_3
\end{bmatrix},\\
A_0&=\begin{bmatrix}
    -g(\alpha_3+f^2g-fg) & -f\,g\\
    \gamma_0(\alpha_3+f^2g-fg) & f\,\gamma_0
\end{bmatrix},
\end{align*}
with
\begin{align*}
    \gamma_0&=\alpha_0(\alpha_0+\alpha_1)-(t+2\alpha_0+\alpha_1-\alpha_3)g+fg(2\alpha_0+\alpha_1+fg-g),\\
    \gamma_1&=\alpha_0(\alpha_0+\alpha_1)(f-1)+(\alpha_3+f^2g-fg)(2\alpha_0+\alpha_1+fg-g).
\end{align*}
In particular, $A(z)$ is a matrix polynomial in $z$ with coefficient entries from $\mathbb{Z}[f,g,\alpha_1,\alpha_2,\alpha_3,t]$.
The corresponding degree one matrix polynomial $B(z)$ in \eqref{eq:laxpvNaB} is given by
\begin{equation*}
    B(z)=\frac{1}{t}\begin{bmatrix}
 -\frac{1}{2}(t+2\alpha_0+\alpha_1-\alpha_3) & 1 \\
 \gamma_1+tfg & \frac{1}{2}(t+2\alpha_0+\alpha_1-\alpha_3) \\
\end{bmatrix}+z\begin{bmatrix}
 0 & 0 \\
 1 & 0\\
\end{bmatrix}.
\end{equation*}
Compatibility of the Lax pair \eqref{eq:laxpvNa} is equivalent to $\Pfive^\text{Na}$.

\begin{remark}\label{rem:nogaugew}
Unlike the Lax pairs for $\Psix$ and $\Pfive^\text{KNY}$, the Lax pair for $\Pfive^\text{Na}$ above has no auxiliary variable $w$. This is because the constant coefficient of the top-left entry of $A(z)$, $(A_1)_{11}$, which plays the role of the auxiliary variable in \eqref{eq:laxpvNa}, can always be normalised to equal $0$ through the $\mathbb{A}_\mathbb{C}^1$-action on $A(z)$ induced by
    \begin{equation}\label{eq:pivgaugefreedom}
     Y(z)\mapsto \widetilde{Y}(z)=R\,Y(z),\quad    A(z)\mapsto \widetilde{A}(z)=RA(z)R^{-1},\qquad R=\begin{bmatrix}
            1 &0\\
            r & 1
        \end{bmatrix},\quad r\in \mathbb{A}_\mathbb{C}^1.
    \end{equation}
Enforcing condition (iv) above, this gauge freedom is completely rigidified. In comparison, the auxiliary $w$ is a necessary feature of e.g. the Lax pair of $\Psix$, as it allows the parametrisation to also capture special cases such as $A_{12}(z)\equiv c$ being a constant, by letting $w\rightarrow 0$ with $fg\sim c w^{-1}$ in \eqref{eq:laxpviA}. Such a limit corresponds to one of the movable singularities of $\Psix$, see e.g. the proof of \cite[Lemma 2.5]{joshiroffelsencrystal}, where the correspondence between exceptional lines and special conditions on the coefficient matrix is made completely explicit for the case of $q\Psix$.
\end{remark}

\subsection{The Hamiltonian}Whilst the trace of the product of $A(z)$ evaluated at successive integers,
$$\operatorname{Tr}[A(n-1)\cdot A(n-2)\cdot\ldots\cdot A(1)\cdot A(0)],$$
gives an integral of motion in characteristic $p$ when $n=p$ as in Theorem \ref{thm:gen trace}, the case $n=1$ is related to the Hamiltonian of $\Pfive^\text{KNY}$. Namely,
\begin{equation*}
    H:=\operatorname{Tr}A_0=f(1+fg-g)^2-(\alpha_1+2\alpha_2+2\alpha_3)f(1+fg-g)+(1-\alpha_0)(\alpha_2+\alpha_3)f+\alpha_3(f-1)g-tfg,
\end{equation*}
is the Hamiltonian of $\Pfive^\text{Na}$, that is,
\begin{align*}
    t\,D_tf&=+\frac{\partial H}{\partial g},\\
    t\,D_tg&=-\frac{\partial H}{\partial f},
\end{align*}
are respectively equivalent to the equations for $D_tf$ and $D_tg$ in \eqref{eq:pvNa}.

\subsection{Integrals in finite characteristic}
The first few integrals of motion $\mathcal{I}_p^{\textrm{V}(\text{Na})}$, $p=2,3$, for $\Pfive^\text{Na}$ in Theorem \ref{thm:gen trace}, are given by
\begin{align*}
    \mathcal{I}_2^{\textrm{V}(\text{Na})}=\,&+f^6 g^4+f^2 g^2 (\alpha_1 f+g)^2+ (\alpha_0+\alpha_2+t)fg(f^2g+\alpha_1 f+g)+(\alpha_1+\alpha_3+t)^2f^2 g^2+\alpha_3^2 g^2\\
    &+ \alpha_0^2(\alpha_0+\alpha_1)^2 f^2+ \left(\alpha_1^2+\alpha_1+\alpha_3^2+\alpha_3+t^2\right)f g+(\alpha_0+\alpha_2+t) \left(\alpha_0(\alpha_0+\alpha_1) f+\alpha_3 g\right) \\
    &+\alpha_0\alpha_3(\alpha_0+\alpha_1),\\
     \mathcal{I}_3^{\textrm{V}(\text{Na})}=\,&+f^9 g^6+f^6 g^6+f^6 g^4+ \left(\alpha_2^3+\alpha_3^3- \alpha_1^3-1\right)f^6 g^3+\left(g^2-f^2\right)f^3 g^4 + (\alpha_0-\alpha_1)f^5 g^3\\
     &+ (g-\alpha_3f-tf+\alpha_1^3-\alpha_2^3-t^3+t+1)f^3 g^3+\alpha_1^2 f^4 g^2+f^2 g^4-  (\alpha_0- \alpha_1- t)f^2 g^3\\
     &+\left((\alpha_2-\alpha_0-t)t-\alpha_3^2+1\right)f^3 g^2-\left(\alpha_1^2(\alpha_0+\alpha_1)+ \alpha_2^3+\alpha_3^3-1\right)f^3 g+\alpha_3 f g^3+f^2 g^2\\
     &+\left(
     \alpha_1^2(\alpha_3-\alpha_0)+\alpha_2^3+\alpha_3^2(\alpha_1-\alpha_2+1)
     - \alpha_1+\alpha_2+\alpha_3+1\right)f^2g\\  
     &+\left(\left(\alpha_1^2- \alpha_1 \alpha_2- \alpha_2^2+\alpha_3^2+\alpha_3+1\right)t+ (\alpha_0-\alpha_1)t^2\right)f^2 g+\left((\alpha_0-\alpha_2-t)t- \alpha_1^2+1\right)f g^2 \\
     &-\alpha_0^3(\alpha_2+\alpha_3-1)^3 f^3- \alpha_3^3 g^3+\alpha_0^2(\alpha_0+\alpha_1)^2 f^2 +\alpha_3^2 g^2\\
     &+\left(\alpha_1^2(\alpha_1+\alpha_3-\alpha_0)+\alpha_3^2(\alpha_0-\alpha_1-\alpha_3)+\alpha_1-\alpha_2+1
    + (\alpha_0\alpha_2+\alpha_1\alpha_3-\alpha_0) t + t^3\right)f g\\
     &-\alpha_0(\alpha_2+\alpha_3-1) \left(\alpha_1^2+\alpha_1 \alpha_3+\alpha_1 t- \alpha_2 \alpha_3- \alpha_2 t+\alpha_3+t^2-1\right)f\\
     &+\left(\alpha_1^2+ \alpha_1 \alpha_3+\alpha_1 t-\alpha_2 \alpha_3-\alpha_2 t+ \alpha_3+ t^2-t-1\right)\alpha_3g-\alpha_0\alpha_3 (\alpha_2+\alpha_3-1) (\alpha_1-\alpha_2-t+1).
\end{align*}

under the identification \eqref{eq:NatoKNY} between the dependent variables of $\Pfive^\text{KNY}$ and $\Pfive^\text{Na}$, we find that the corresponding integrals of motion $\mathcal{I}_p^{\textrm{V}(\text{KNY})}$ and $\mathcal{I}_p^{\textrm{V}(\text{Na})}$ are related by
\begin{equation*}
   \mathcal{I}_p^{\textrm{V}(\text{KNY})} +\mathcal{I}_p^{\textrm{V}(\text{Na})}=\mathcal{A}_3(\mathcal{A}_1+\mathcal{A}_2+\mathcal{A}_3)-(\mathcal{A}_2+\mathcal{A}_3)t^p,
\end{equation*}
for $p=2,3,5,7,11$, where we used the notation in \eqref{eq:notationps_reform}.
We expect this relation to hold for all primes $p$.

\section{The Fourth Painlev\'e equation}\label{sect:p4}
\subsection{The nonlinear system}
Fix a field $\mathbf{k}$ and consider the field of rational functions $\mathbf{k}(f,g,\alpha_1,\alpha_2,t)$ with the $\mathbf{k}$-linear derivation $D_t$ specified by
 \begin{equation}\label{eq:piv}
 \Pfour:\quad           \begin{cases}
        D_t f=f(2\hspace{0.3mm}g-f-t)-\alpha_1, &   D_t \alpha_1=0, \\
   D_t \hspace{0.3mm}g=g\hspace{0.3mm}(2f-\hspace{0.3mm}g+t)+\alpha_2, &   D_t \alpha_2=0, \\
        D_t \hspace{0.4mm}t\hspace{0.4mm}=1. & 
    \end{cases}
    \end{equation}
 This system of ODEs is known as the fourth Painlev\'e equation (in system form). Assuming the characteristic of $\mathbf{k}$ is not $2$, eliminating $g$ from equations \eqref{eq:piv}, gives the following scalar form of the fourth Painlev\'e equation for $f$,
\begin{align*}
    D_t^2f=&\frac{(D_tf)^2}{2f}+\frac{3}{2}f^3+2tf^2+(\tfrac{1}{2}t^2-\alpha)f+\frac{\beta}{4f},
\end{align*}
with
\begin{equation*}
    \alpha=1-\alpha_1-2\alpha_2,\quad \beta=-2\alpha_1^2.
\end{equation*}

\subsection{The Lax pair}\label{subsect:linP4}
We consider the following difference-differential Lax pair, obtained by normalising and gauge transforming the Lax pair for $\Pfour$ in \cite{nakazono_lax} and first derived in \cite{adler94},
\begin{subequations}\label{eq:laxpiv}
\begin{align}
    Y(z+1)=A(z)Y(z),\label{eq:laxpivA}\\
    D_t Y(z)=B(z)Y(z),\label{eq:laxpivB}
\end{align}
\end{subequations}
where the coefficient matrix $A(z)$ is a $2\times 2$ matrix polynomial of degree two in $z$,
\begin{equation}\label{eq:pivAseries}
        A(z)=A_0+z A_1+z^2A_2,
    \end{equation}
    characterised by the following four properties.
\begin{enumerate}[label=(\roman*)]
\item The leading order coefficient is given by
\begin{equation*}
    \qquad A_2=\begin{bmatrix}
            0 & 0\\
            1 & 0
        \end{bmatrix}.
\end{equation*}
    \item The trace of $A(z)$ has leading order behaviour
   \begin{equation}\label{eq:traceinfpiv}
    \begin{aligned}
     \operatorname{Tr}A(z)=t\, z+\mathcal{O}(1),
    \end{aligned}
   \end{equation}
    as $z\rightarrow \infty$.
    \item The determinant of $A(z)$ equals
\begin{equation}\label{eq:coefmatrixdetpiv}
    |A(z)|=-z(z-\alpha_1)(z+\alpha_2).
\end{equation}
\item The top-left entry of $A_1$ is zero, $(A_1)_{11}=0$. 
\end{enumerate}
In other words, the parameters $\{\alpha_1,\alpha_2\}$ can be thought of as controlling the locations of the zeros of the determinant of $A(z)$, whereas $t$ characterises the singularity at $z=\infty$ of equation \eqref{eq:laxpivA}.

We use standard coordinates $\{f,g\}$ on the coefficient matrix, defined by 
\begin{equation*}
    A_{12}(z)=z-f\,g,\qquad A_{11}(z)|_{z=fg}\,=g\,(f\,g-\alpha_1),
\end{equation*}
where $A_{12}(z)$ denotes the $(1,2)$-entry of $A(z)$ and $A_{11}(z)|_{z=fg}$ denotes the $(1,1)$-entry of $A(z)$ evaluated at $z=fg$. 
These coordinates combined with properties (i)--(iii) result in $A_1$ and $A_0$ in \eqref{eq:pivAseries} to be given by
\begin{align*}
A_1&=\begin{bmatrix}
 0 & 1\\
f\,g-\alpha_1+\alpha_2 & t
\end{bmatrix},\\
A_0&=\begin{bmatrix}
g(f\,g-\alpha_1) & -f\,g\\
(f\,g-\alpha_1)(\alpha_2+(f+t)g) & -f(\alpha_2+(f+t)g)
\end{bmatrix}.
\end{align*}
In particular, $A(z)$ is a matrix polynomial in $z$ with coefficient entries from $\mathbb{Z}[f,g,\alpha_1,\alpha_2,t]$.
The corresponding degree one matrix polynomial $B(z)$ in \eqref{eq:laxpivB} is given by,
\begin{equation*}
    B(z)=\begin{bmatrix}
 -\frac{1}{2}t & 1\\
 2\,f\,g-\alpha_1+\alpha_2 & \frac{1}{2}t
\end{bmatrix}+z\begin{bmatrix}
 0 & 0 \\
 1 & 0\\
\end{bmatrix}.
\end{equation*}
Compatibility of the Lax pair \eqref{eq:laxpiv} is equivalent to $\Pfour$.

This Lax pair has no auxiliary variable $w$ for the same reason that $\Pfive^\text{Na}$ does not have one, see Remark \ref{rem:nogaugew}.

\subsection{The Hamiltonian}Whilst the trace of the product of $A(z)$ evaluated at successive integers,
$$\operatorname{Tr}[A(n-1)\cdot A(n-2)\cdot\ldots\cdot A(1)\cdot A(0)],$$
gives an integral of motion in characteristic $p$ when $n=p$ as in Theorem \ref{thm:gen trace}, the case $n=1$ is related to the Hamiltonian of the fourth Painlev\'e equation. Namely,
\begin{equation*}
    H:=\operatorname{Tr}A_0=-f\,g(t+f-g)-\alpha_1g-\alpha_2f,
\end{equation*}
is the Hamiltonian of $\Pfour$, that is,
\begin{align*}
    D_tf&=+\frac{\partial H}{\partial g},\\
    D_tg&=-\frac{\partial H}{\partial f},
\end{align*}
are respectively equivalent to the equations for $D_tf$ and $D_tg$ in \eqref{eq:piv}.

\subsection{Integrals in finite characteristic}

The first few integrals of motion $\mathcal{I}_p^{\textrm{IV}}$, $p=2,3,5$, for $\Pfour$ in Theorem \ref{thm:gen trace}, are given by
\begin{align*}
    \mathcal{I}_2^{\textrm{IV}}=\,& 
   + f^2 g^2(t^2+f^2+g^2)+t (t+f+g)f g+\alpha_2^2 f^2+\alpha_1^2 g^2+\alpha_2 t f+\alpha_1 t g+f g+\alpha_1 \alpha_2, 
    \\
    \mathcal{I}_3^{\textrm{IV}}=\,&-f^3 g^3 (t^3+f^3-g^3)
 +t^2 (t+f-g)  f g 
  -(\alpha _1-\alpha _2)  (f-g)f g+\alpha _1 t   (t-f)g+\alpha _2 t f  (g+t)- (f+g)f g\\
  &-\alpha _2^3 f^3-\alpha _1^3 g^3
  +\alpha _2(\alpha _2-1)  f-\alpha _1(\alpha _1-1)  g+\alpha _1\alpha _2  (g-f+t),\\
      \mathcal{I}_5^{\textrm{IV}}=\,&-f^5 g^5 (t^5+f^5-g^5)+t f g (t+f-g) (t^3+2 f g(t+f-g))
      -\alpha_1^5 g^5-\alpha_2^5 f^5-2 f^3 g^2-2 f^2 g^3\\
      &+t^4 (\alpha_2 f+\alpha_1 g)+2 t^3(\alpha_1-\alpha_2) f g +t^2 (2 \alpha_1+2 \alpha_2-1)(f+g) f g -t f g (f-g) (\alpha_2 f+\alpha_1 g)\\
      &+t^2 \left(2 \alpha_1^2 g-2 \alpha_2^2 f+2 \alpha_1 \alpha_2 f-2 \alpha_1 \alpha_2 g+\alpha_1 g-\alpha_2 f\right)+t \left(2 \alpha_2^2 f^2+2 \alpha_1^2 g^2+(\alpha_1^2+\alpha_2^2) f g\right)\\
      &+f g \left(\alpha_1^2+\alpha_1 \alpha_2-\alpha_1+\alpha_2^2-\alpha_2-2\right) (f-g)+\alpha_2 f \left(\alpha_1^2+\alpha_2^2+\alpha_1 \alpha_2-\alpha_1+\alpha_2-2\right)\\
      &+\alpha_1 g \left(\alpha_1^2+\alpha_2^2+\alpha_1 \alpha_2+\alpha_1-\alpha_2-2\right)+\alpha_1 \alpha_2 t^3-2 t \alpha_1 \alpha_2 (\alpha_1-\alpha_2).
\end{align*}

\section{The Third Painlev\'e equation of $D_6$ type} \label{sect:p3d6}
\subsection{The nonlinear system}
Fix a field $\mathbf{k}$ and consider the field of rational functions $\mathbf{k}(f,g,\alpha_1,\alpha_2,t)$ with the $\mathbf{k}$-linear derivation $D_t$ specified by
 \begin{equation}\label{eq:piii}
 \Pthree^{D_6}:\quad           \begin{cases}
       t\, D_t f=\hspace{0.4mm}t\hspace{0.4mm}+(\alpha_1+\alpha_2)f+2f^2g-f^2, &   D_t \alpha_1=0, \\
  t\, D_t \hspace{0.2mm}g\hspace{0.1mm}=\hspace{-0.2mm} \alpha_2\hspace{-0.2mm}-(\alpha_1+\alpha_2)g-2fg^2+2 fg, &   D_t \alpha_2=0, \\
        D_t \hspace{0.4mm}t\hspace{0.4mm}=1. &    \\ 
    \end{cases}
    \end{equation}
This system of ODEs is known as the third Painlev\'e equation of $D_6$ type (in system form).

Assuming the characteristic of $\mathbf{k}$ is not $2$, eliminating $g$ from equations \eqref{eq:piii}, gives the following scalar form of the third Painlev\'e equation for $f$,
\begin{equation*}
    D_t^2f=\frac{(D_tf)^2}{f}-\frac{1}{t}D_tf+\frac{1}{t^2}f^3+\frac{\alpha_2-\alpha_1}{t^2}f^2+\frac{1-\alpha_1-\alpha_2}{t}-\frac{1}{f}.
\end{equation*}

\subsection{The Lax pair} \label{subsect:linP3d6}
We consider the following difference-differential Lax pair, which we obtained by matching up an ansatz system form with the scalar Lax pair in \cite{KNY2017} for $\Pthree^{D_6}$, though we note a system form is also given in \cite[\S 7]{adler94},
\begin{subequations}\label{eq:laxpiii}
\begin{align}
    Y(z+1)=A(z)Y(z),\label{eq:laxpiiiA}\\
    D_t Y(z)=B(z)Y(z),\label{eq:laxpiiiB}
\end{align}
\end{subequations}
where the coefficient matrix $A(z)$ is a $2\times 2$ matrix polynomial of degree two in $z$,
\begin{equation}\label{eq:piiiAseries}
        A(z)=A_0+z A_1+z^2A_2,
    \end{equation}
    characterised by the following three properties.
\begin{enumerate}[label=(\roman*)]
\item The leading order coefficient is given by
\begin{equation*}
    \qquad A_2=\begin{bmatrix}
            1 & 0\\
            0 & 0
        \end{bmatrix}.
\end{equation*}
    \item The trace of $A(z)$ has leading order behaviour
   \begin{equation}\label{eq:traceinfpiii}
    \begin{aligned}
     \operatorname{Tr}A(z)=z^2+(\alpha_1+\alpha_2-1)z+\mathcal{O}(1),
    \end{aligned}
   \end{equation}
    as $z\rightarrow \infty$.
    \item The determinant of $A(z)$ equals
\begin{equation}\label{eq:coefmatrixdetpiii}
    |A(z)|=-t\,z(z+\alpha_2).
\end{equation}
\end{enumerate}
In other words, the parameter $\alpha_2$ can be thought of as controlling the location of the zero of the determinant of $A(z)$, whereas $\{t,\alpha_1\}$ characterises the singularity at $z=\infty$ of equation \eqref{eq:laxpiiiA}. We use standard coordinates $\{f,g,w\}$ on the coefficient matrix, defined by 
\begin{equation*}
    A_{12}(z)=w(z-f\,g),\qquad A_{11}(z)|_{z=fg}\,=-t\,g,
\end{equation*}
where $A_{12}(z)$ denotes the $(1,2)$-entry of $A(z)$ and $A_{11}(z)|_{z=fg}$ denotes the $(1,1)$-entry of $A(z)$ evaluated at $z=fg$. 
These coordinates combined with properties (i)--(iii) result in $A_1$ and $A_0$ in \eqref{eq:piiiAseries} to be given by
\begin{align*}
A_1&=\begin{bmatrix}
 \alpha_1+\alpha_2-1 & w\\
 [t+f(\alpha_2+fg)]\,w^{-1} & 0
\end{bmatrix},\\
A_0&=\begin{bmatrix}
-gh & -f\, g\, w\\
(\alpha_2+fg)h\,w^{-1}& f(\alpha_2+fg)
\end{bmatrix},
\end{align*}
where,
\[ h = t+(\alpha_1+\alpha_2-1)f+f^2g.\]
In particular, $A(z)$ is a matrix polynomial in $z$ with coefficient entries from $\mathbb{Z}[f,g,\alpha_1,\alpha_2,t,w,w^{-1}]$.

The corresponding degree one matrix polynomial $B(z)$ in \eqref{eq:laxpiiiB} is given by,
\begin{equation*}
    B(z)=-\frac{1}{t}\begin{bmatrix}
 0 & w \\
 [t+f(\alpha_2+fg)]\,w^{-1} & 0 \\
\end{bmatrix}+\frac{z}{t}\begin{bmatrix}
 0 & 0 \\
 0 & 1\\
\end{bmatrix}.
\end{equation*}

Compatibility of the Lax pair \eqref{eq:laxpiii} is equivalent to $\Pthree^{D_6}$ and the following ODE for the auxiliary variable $w$,
\begin{equation*}
    \frac{D_tw}{w}=\frac{f\,g+\alpha_1+\alpha_2-1}{t}.
\end{equation*}

\subsection{The Hamiltonian}Whilst the trace of the product of $A(z)$ evaluated at successive integers,
$$\operatorname{Tr}[A(n-1)\cdot A(n-2)\cdot\ldots\cdot A(1)\cdot A(0)],$$
gives an integral of motion in characteristic $p$ when $n=p$ as in Theorem \ref{thm:gen trace}, the case $n=1$ is related to the Hamiltonian of $\Pthree^{D_6}$. Namely,
\begin{equation*}
    H:=f\,g-\operatorname{Tr}A_0=t\,g+\alpha_1 f g+\alpha_2f(g-1)+f^2g(g-1),
\end{equation*}
is the Hamiltonian of $\Pthree^{D_6}$, that is,
\begin{align*}
    t\,D_tf&=+\frac{\partial H}{\partial g},\\
    t\,D_tg&=-\frac{\partial H}{\partial f},
\end{align*}
are respectively equivalent to the equations for $D_tf$ and $D_tg$ in \eqref{eq:piii}.

\subsection{Integrals in finite characteristic}
Below are the integrals of motion $\mathcal{I}_p^{(\rm{IIID}_6)}$, $p=2,3$, as defined in Theorem \ref{thm:gen trace}.
\begin{align*}
\mathcal{I}_2^{(\rm{IIID}_6)} =&\,+f^4 (g^2+g^4)+f^2 (\alpha_{2}^2+g (1+\alpha_{1}+\alpha_{2})+g^2 (1+\alpha_{1}+\alpha_{1}^2+\alpha_{2}+\alpha_{2}^2))\\
&\,+f (\alpha_{2} (1+\alpha_{1}+\alpha_{2})+g (\alpha_{1}+\alpha_{1}^2+\alpha_{2}+\alpha_{2}^2)) \\
&\,+t (g^2 t+\alpha_{2}+g (\alpha_{1}+\alpha_{2})) ,\\
\mathcal{I}_3^{(\rm{IIID}_6)} =&\,+f^6 (g^3+2 g^6)+2 f^4 g^2 (1+g+g^2)+f^3 (g \alpha_{2}+\alpha_{2}^3+g^2 (2 \alpha_{1}+\alpha_{2})+g^3 (\alpha_{1}+2 \alpha_{1}^3+\alpha_{2}+2 \alpha_{2}^3))\\
&\,+f^2 (2 g^2+g^3 t+2 \alpha_{2}^2+g (1+2 t+2 \alpha_{1}^2+\alpha_{2}^2))\\
&\,+f (g^2 t (2+\alpha_{1}+\alpha_{2})+\alpha_{2} (1+2 t+2 \alpha_{1}^2+\alpha_{1} \alpha_{2}+2 \alpha_{2}^2)+g (t+2 \alpha_{1}+t \alpha_{1}+\alpha_{1}^3+2 \alpha_{2}+\alpha_{2}^3))\\
&\,+t (2 g^2 t+2 g^3 t^2+2 \alpha_{2} (1+\alpha_{1}+\alpha_{2})+g (t+\alpha_{1}+\alpha_{1}^2+\alpha_{2}+2 \alpha_{1} \alpha_{2}+\alpha_{2}^2)) .
\end{align*}

\section{The Third Painlev\'e equation of $D_7$ type}\label{sect:p3d7}
\subsection{The nonlinear system}
Fix a field $\mathbf{k}$ and consider the field of rational functions $\mathbf{k}(f,g,\alpha_1,t)$ with the $\mathbf{k}$-linear derivation $D_t$ specified by
 \begin{equation}\label{eq:piiid7}
 \Pthree^{D_7}:\quad           \begin{cases}
       t\, D_t f=+t+f(\alpha_1+2\,f\,g), &   D_t \alpha_1=0, \\
  t\, D_t \hspace{0.2mm}g\hspace{0.1mm}=-1-g(\alpha_1+2\,f\,g), &   D_t \hspace{0.4mm}t\hspace{0.4mm}=1. 
    \end{cases}
    \end{equation}
This system of ODEs is known as the third Painlev\'e equation of $D_7$ type (in system form). Assuming the characteristic of $\mathbf{k}$ is not $2$, eliminating $g$ from equations \eqref{eq:piiid7}, gives the following scalar form of the third Painlev\'e equation of $D_7$ type for $f$,
\begin{equation*}
    D_t^2f=\frac{(D_tf)^2}{f}-\frac{1}{t}D_tf-\frac{2}{t^2}f^2+\frac{1-\alpha_1}{t}-\frac{1}{f}.
\end{equation*}

\subsection{The Lax pair} \label{subsect:linP3d7}
We consider the following difference-differential Lax pair, which we obtained by matching up an ansatz system form with the scalar Lax pair in \cite{KNY2017} for $\Pthree^{D_7}$,
\begin{subequations}\label{eq:laxpiiid7}
\begin{align}
    Y(z+1)=A(z)Y(z),\label{eq:laxpiiid7A}\\
    D_t Y(z)=B(z)Y(z),\label{eq:laxpiiid7B}
\end{align}
\end{subequations}
where the coefficient matrix $A(z)$ is a $2\times 2$ matrix polynomial of degree two in $z$,
\begin{equation}\label{eq:piiid7Aseries}
        A(z)=A_0+z A_1+z^2A_2,
    \end{equation}
    characterised by the following three properties.
\begin{enumerate}[label=(\roman*)]
\item The leading order coefficient is given by
\begin{equation*}
    \qquad A_2=\begin{bmatrix}
            1 & 0\\
            0 & 0
        \end{bmatrix}.
\end{equation*}
    \item The trace of $A(z)$ has leading order behaviour
   \begin{equation}\label{eq:traceinfpiiid7}
    \begin{aligned}
     \operatorname{Tr}A(z)=z^2+(\alpha_1-1)z+\mathcal{O}(1),
    \end{aligned}
   \end{equation}
    as $z\rightarrow \infty$.
    \item The determinant of $A(z)$ equals
\begin{equation}\label{eq:coefmatrixdetpiiid7}
    |A(z)|=t\,z.
\end{equation}
\end{enumerate}
In other words, the parameters $\{t,\alpha_1\}$ characterise the singularity at $z=\infty$ of equation \eqref{eq:laxpiiid7A}. We use standard coordinates $\{f,g,w\}$ on the coefficient matrix, defined by 
\begin{equation*}
    A_{12}(z)=w(z-f\,g),\qquad A_{11}(z)|_{z=fg}\,=-t\,g,
\end{equation*}
where $A_{12}(z)$ denotes the $(1,2)$-entry of $A(z)$ and $A_{11}(z)|_{z=fg}$ denotes the $(1,1)$-entry of $A(z)$ evaluated at $z=fg$. 
These coordinates combined with properties (i)--(iii) result in $A_1$ and $A_0$ in \eqref{eq:piiid7Aseries} to be given by
\begin{align*}
A_1&=\begin{bmatrix}
 \alpha_1-1 & w\\
 -f\,w^{-1} & 0
\end{bmatrix},\\
A_0&=\begin{bmatrix}
-g\,[t+(\alpha_1-1)f+f^2g] & -f\, g\, w\\
-[t+(\alpha_1-1)f+f^2g]\,w^{-1}& -f
\end{bmatrix}.
\end{align*}
In particular, $A(z)$ is a matrix polynomial in $z$ with coefficient entries from $\mathbb{Z}[f,g,\alpha_1,t,w,w^{-1}]$.

The corresponding degree one matrix polynomial $B(z)$ in \eqref{eq:laxpiiid7B} is given by,
\begin{equation*}
    B(z)=\frac{1}{t}\begin{bmatrix}
 0 & -w \\
 f\,w^{-1} & 0 \\
\end{bmatrix}+\frac{z}{t}\begin{bmatrix}
 0 & 0 \\
 0 & 1\\
\end{bmatrix}.
\end{equation*}

Compatibility of the Lax pair \eqref{eq:laxpiiid7} is equivalent to $\Pthree^{D_7}$ and the following ODE for the auxiliary variable $w$,
\begin{equation*}
    \frac{D_tw}{w}=\frac{f\,g+\alpha_1-1}{t}.
\end{equation*}

\subsection{The Hamiltonian}Whilst the trace of the product of $A(z)$ evaluated at successive integers,
$$\operatorname{Tr}[A(n-1)\cdot A(n-2)\cdot\ldots\cdot A(1)\cdot A(0)],$$
gives an integral of motion in characteristic $p$ when $n=p$ as in Theorem \ref{thm:gen trace}, the case $n=1$ is related to the Hamiltonian of $\Pthree^{D_7}$. Namely,
\begin{equation*}
    H:=f\,g-\operatorname{Tr}A_0=f+t\,g+f\,g(\alpha_1+f\,g),
\end{equation*}
is the Hamiltonian of $\Pthree^{D_7}$, that is,
\begin{align*}
    t\,D_tf&=+\frac{\partial H}{\partial g},\\
    t\,D_tg&=-\frac{\partial H}{\partial f},
\end{align*}
are respectively equivalent to the equations for $D_tf$ and $D_tg$ in \eqref{eq:piiid7}.

\subsection{Integrals in finite characteristic}
Below are the integrals of motion $\mathcal{I}_p^{(\rm{IIID}_7)}$, $p=2,3,5$, as defined in Theorem \ref{thm:gen trace}.
\begin{align*}
\mathcal{I}_2^{(\rm{IIID}_7)} =\,&+f^4 g^4+f^2g^2(\alpha _1^2+\alpha _1+2)+f^2+\left(\alpha _1+1\right) f \left(\alpha _1 g+1\right)
+t \left(t g^2+\alpha _1 g+1\right),\\
\mathcal{I}_3^{(\rm{IIID}_7)} =\,&-f^6 g^6-f^4 g^4-\left(\alpha_1^3-\alpha_1\right) f^3 g^3+f^3 \left(g^2-1\right)+tf^2 g^3 -t^3g^3-f^2 g^2-f^2+\alpha_1 f g (f+t g+\alpha_1^2-1)\\
&-t  (t+f)g^2+t f g+\left(\alpha_1^2-1\right) f+\alpha_1 (\alpha_1+1) t g
+(\alpha_1+1) t,\\
 \mathcal{I}_5^{(\rm{IIID}_7)} =\,&
 -f^{10} g^{10}+2 f^6 g^6 
 -\left(\alpha_1^5-\alpha_1\right) f^5 g^5 +f^4 g^4 (t g+f)
 +2 f^3 g^3 (\alpha_1 f+(\alpha_1-1) t g)+f^2 g^2 (t g+f)^2\\
 & -f^2 g^2 \left(\alpha_1^2 f+(\alpha_1-1)^2 t g\right)-f^5-t^5 g^5+f g \left(\alpha_1 f^2+(\alpha_1-1) t^2 g^2+f g (\alpha_1 t+2 t-1)\right)\\
 &+3 \alpha_1^3 t f g^2+3 \alpha_1^3 f^2 g+\alpha_1^2 t f g^2- \alpha_1 t f g^2+2 t^3 g^3+t^2 f g^2+t f^2 g+2 t f g^2+2 f^3+\alpha_1^5 f g- \alpha_1^2 t^2 g^2\\
 &+2 \alpha_1^2 t f g- \alpha_1^2 f^2+2 \alpha_1 t^2 g^2+3 \alpha_1 t f g- \alpha_1 f g- t^2 g^2- t f g+\alpha_1^4 t g+\alpha_1^4 f+\alpha_1^3 t g+\alpha_1^2 t g\\
 &- \alpha_1 t^2 g- \alpha_1 t f+\alpha_1 t g+2 t^2 g- t f- f+t+\alpha_1 t+\alpha_1^2 t+\alpha_1^3 t- t^2.
\end{align*}

\section{The Third Painlev\'e equation of $D_8$ type}\label{sect:p3d8}
\subsection{The nonlinear system}
Fix a field $\mathbf{k}$ and consider the field of rational functions $\mathbf{k}(f,g,t)$ with the $\mathbf{k}$-linear derivation $D_t$ specified by
 \begin{equation}\label{eq:piiid8}
 \Pthree^{D_8}:\quad           \begin{cases}
       t\, D_t f=f(1+2\,f\,g), &   D_t \hspace{0.4mm}t\hspace{0.4mm}=1, \\
  \displaystyle t\, D_t \hspace{0.2mm}g\hspace{0.1mm}=\frac{t}{f^2}-1-g(1+2\,f\,g). &   
    \end{cases}
    \end{equation}
This system of ODEs is known as the third Painlev\'e equation of $D_8$ type (in system form). Assuming the characteristic of $\mathbf{k}$ is not $2$, eliminating $g$ from equations \eqref{eq:piiid8}, gives the following scalar form of the third Painlev\'e equation of $D_8$ type for $f$,
\begin{equation*}
    D_t^2f=\frac{(D_tf)^2}{f}-\frac{1}{t}D_tf-\frac{2}{t^2}f^2+\frac{2}{t}.
\end{equation*}

\subsection{The Lax pair} \label{subsect:linP3d8}
We consider the following difference-differential Lax pair, which we obtained by matching up an ansatz system form with the scalar Lax pair in \cite{KNY2017} for $\Pthree^{D_8}$, though we note a system form is also given in \cite[\S 6]{adler94},
\begin{subequations}\label{eq:laxpiiid8}
\begin{align}
    Y(z+1)=A(z)Y(z),\label{eq:laxpiiid8A}\\
    D_t Y(z)=B(z)Y(z),\label{eq:laxpiiid8B}
\end{align}
\end{subequations}
where the coefficient matrix $A(z)$ is a $2\times 2$ matrix polynomial of degree two in $z$,
\begin{equation}\label{eq:piiid8Aseries}
        A(z)=A_0+z A_1+z^2A_2,
    \end{equation}
    characterised by the following three properties.
\begin{enumerate}[label=(\roman*)]
\item The leading order coefficient is given by
\begin{equation*}
    \qquad A_2=\begin{bmatrix}
            1 & 0\\
            0 & 0
        \end{bmatrix}.
\end{equation*}
    \item The trace of $A(z)$ has leading order behaviour
   \begin{equation}\label{eq:traceinfpiiid8}
    \begin{aligned}
     \operatorname{Tr}A(z)=z^2+\mathcal{O}(1),
    \end{aligned}
   \end{equation}
    as $z\rightarrow \infty$.
    \item The determinant of $A(z)$ equals
\begin{equation}\label{eq:coefmatrixdetpiiid8}
    |A(z)|=t.
\end{equation}
\end{enumerate}
In other words, the parameter $t$ completely characterises the singularity at $z=\infty$ of equation \eqref{eq:laxpiiid8A}. We use standard coordinates $\{f,g,w\}$ on the coefficient matrix, defined by 
\begin{equation*}
    A_{12}(z)=w(z-f\,g),\qquad A_{22}(z)|_{z=fg}\,=-f,
\end{equation*}
where $A_{12}(z)$ denotes the $(1,2)$-entry of $A(z)$ and $A_{22}(z)|_{z=fg}$ denotes the $(2,2)$-entry of $A(z)$ evaluated at $z=fg$. 
These coordinates combined with properties (i)--(iii) result in $A_1$ and $A_0$ in \eqref{eq:piiid8Aseries} to be given by
\begin{align*}
A_1&=\begin{bmatrix}
 0 & w\\
 -f\,w^{-1} & 0
\end{bmatrix},\\
A_0&=-\begin{bmatrix}
f^2g^2+t\,f^{-1} & f\, g\, w\\
f^2g\,w^{-1}& f
\end{bmatrix}.
\end{align*}
In particular, $A(z)$ is a matrix polynomial in $z$ with coefficient entries from $\mathbb{Z}[f,f^{-1},g,t,w,w^{-1}]$.

The corresponding degree one matrix polynomial $B(z)$ in \eqref{eq:laxpiiid8B} is given by,
\begin{equation*}
    B(z)=\frac{1}{t}\begin{bmatrix}
 0 & -w \\
 f\,w^{-1} & 0 \\
\end{bmatrix}+\frac{z}{t}\begin{bmatrix}
 0 & 0 \\
 0 & 1\\
\end{bmatrix}.
\end{equation*}

Compatibility of the Lax pair \eqref{eq:laxpiiid8} is equivalent to $\Pthree^{D_8}$ and the following ODE for the auxiliary variable $w$,
\begin{equation*}
    \frac{D_tw}{w}=\frac{f\,g}{t}.
\end{equation*}

\subsection{The Hamiltonian}Whilst the trace of the product of $A(z)$ evaluated at successive integers,
$$\operatorname{Tr}[A(n-1)\cdot A(n-2)\cdot\ldots\cdot A(1)\cdot A(0)],$$
gives an integral of motion in characteristic $p$ when $n=p$ as in Theorem \ref{thm:gen trace}, the case $n=1$ is related to the Hamiltonian of $\Pthree^{D_8}$. Namely,
\begin{equation*}
    H:=f\,g-\operatorname{Tr}A_0=f+t\,f^{-1}+f\,g+f^2g^2,
\end{equation*}
is the Hamiltonian of $\Pthree^{D_8}$, that is,
\begin{align*}
    t\,D_tf&=+\frac{\partial H}{\partial g},\\
    t\,D_tg&=-\frac{\partial H}{\partial f},
\end{align*}
are respectively equivalent to the equations for $D_tf$ and $D_tg$ in \eqref{eq:piiid8}.

\subsection{Integrals in finite characteristic}
Below are the integrals of motion $\mathcal{I}_p^{(\rm{IIID}_8)}$, $p=2,3,5$, as defined in Theorem \ref{thm:gen trace}.
\begin{align*}
\mathcal{I}_2^{(\rm{IIID}_8)} =&\,+ f^4 g^4+f^2 g^2+g^2+\frac{t}{f}+ \frac{t^2}{f^2},\\
\mathcal{I}_3^{(\rm{IIID}_8)} =&\, -f^6 g^6-f^4 g^4+f^3 \left(g^2-1\right)-f^2 \left(g^2-g+1\right)+t f g^2-t-\frac{t}{f}-\frac{ t^2}{f^2}-\frac{ t^3}{f^3},\\
 \mathcal{I}_5^{(\rm{IIID}_8)} =&\, -f^{10} g^{10}+2 f^6 g^6+f^5 \left(g^4-1\right)+f^4g^2 \left(2g+1\right)+f^3 \left(t g^4-g^2+ g+2\right)-f^2 \left( g^2(1+t)+2 g+1\right)\\
 &\,+2tf (g-1)+t^2 g^2+2 t-\frac{t(t+2)}{f}+2\frac{ t^3}{f^3}-\frac{t^5}{f^5}.
\end{align*}

\section{Conclusion}\label{sec:conclusion}
We collected all rank 2 difference-differential Lax pairs in the literature for the  Painlev\'e equations 
\begin{equation*}
    P_J,\quad J\in\{\rm{VI}, \rm{V}, \rm{IV}, \rm{III}^{D_6},\rm{III}^{D_7},\rm{III}^{D_8}\},
\end{equation*}
and put each in degree 2 system form. From these Lax pairs we obtained a countable list of integrals of motion for each Painlev\'e equation,
indexed by the primes, so that the entry indexed by prime $p$ is an integral of the relevant Painlev\'e equation in characteristic $p$.

Several natural directions for future work remain.
A rank 3 difference-differential Lax pair for $\Ptwo$ is available in the literature \cite{adler94}, and it would be interesting to apply our method in that setting. By contrast, the existence of a difference-differential Lax pair for $\Pone$ appears to be an open problem and correspondingly it is unclear whether $\Pone$ admits non-trivial integrals of motion in characteristic $p$. Another fundamental question is whether integrals can also be extracted from purely differential Lax pairs. It also remains to be determined whether the integrals found in this paper are minimal or can be further simplified.
Importantly, these integrals provide new tools for the arithmetic study of Painlev\'e dynamics in finite characteristic. In particular, they may be useful for the study of Painlev\'e dynamics over finite fields, potentially leading to applications in cryptography.

\appendix

\section{Proof of Lemma \ref{lem:trace}}\label{sec:trac_finite}
The goal of this section is to prove Lemma \ref{lem:trace}. The proof requires some preparatory lemmas and is given after Lemma \ref{lem:invpol}. First, recall the following well-known power sum congruence modulo $p$, see e.g. \cite{macmillansondow},
\begin{equation}\label{eq:powersumcongruence}
\sum_{n=0}^{p-1}n^k\, \equiv \begin{cases}
    -1 &\text{if $p-1\mid k$},\\
    0 &\text{if $p-1 \nmid k$}.
\end{cases}
\end{equation}

\begin{lemma}\label{lemma:z+n}
For any odd prime $p$ and integer $0\leq N\leq 2p-1$, we have the following identities in  $\mathbb{F}_p[z]$,
\begin{equation}\label{eq:sumidentity}
\sum_{n=0}^{p-1}(z+n)^N\, = \begin{cases}
0 &\text{if $0\leq N < p-1$},\\
-1 &\text{if $N=p-1 $},\\
0 &\text{if $p\leq N < 2p-2 $},\\
-1 &\text{if $ N = 2p-2 $},\\
z-z^p &\text{if $N=2p-1 $}.
\end{cases}
\end{equation}
\end{lemma}
\begin{proof}
    We begin by expanding the summands on the left-hand side of equation \eqref{eq:sumidentity} using the binomial theorem,
    \begin{eqnarray}
        \sum_{n=0}^{p-1}(z+n)^N &=& \sum_{n=0}^{p-1} \sum_{k=0}^N \binom{N}{k}z^{N-k}n^k\nonumber \\
        &=& \sum_{k=0}^N \binom{N}{k}z^{N-k} \sum_{n=0}^{p-1} n^k\nonumber\\
        &=&-\binom{N}{p-1}z^{N-p+1}-\binom{N}{2p-2}z^{N-2p+2},\label{eq:2term}
    \end{eqnarray}
  where we applied \eqref{eq:powersumcongruence} to the inner summations in the last equality, with by definition $\binom{N}{k}=0$ when $0\leq N<k$, and we used that $p$ is an odd prime so that
$k=l(p-1)\leq N$ implies $l\leq 2$ for $l\in\mathbb{Z}$.
  Next, we apply the following identities in characteristic $p$,
\begin{equation*}
\binom{N}{p-1}\, = \begin{cases}
0 &\text{if $0\leq N < p-1$},\\
1 &\text{if $N=p-1 $},\\
0 &\text{if $p\leq N < 2p-2 $},\\
0 &\text{if $ N = 2p-2 $},\\
1 &\text{if $N=2p-1 $},
\end{cases},\qquad
\binom{N}{2p-2}\, = \begin{cases}
0 &\text{if $0\leq N < p-1$},\\
0 &\text{if $N=p-1 $},\\
0 &\text{if $p\leq N < 2p-2 $},\\
1 &\text{if $ N = 2p-2 $},\\
-1 &\text{if $N=2p-1 $},
\end{cases}
\end{equation*}
to obtain the formula in the lemma.   
\end{proof}

\begin{remark}
    When $p=2$, the right-hand side of equation \eqref{eq:2term}  gets an additional term
    \begin{equation*}
        -\binom{N}{3p-3}z^{N-3p+3},
    \end{equation*}
    which affects the result of the calculation when $N=2p-1=3=3p-3$, so that the corresponding sum on the left-hand side of \eqref{eq:sumidentity} evaluates to $1+z+z^2$ for $p=2$ and $N=3$.
\end{remark}

\begin{lemma}\label{lem:invpol} For any prime $p$,
    if a polynomial $f(z)\in \mathbb{F}_p[z]$ satisfies $f(z+1)=f(z)$, then it is a polynomial in $z^p-z$, i.e. $f(z)\in\mathbb{F}_p[z^p-z]$.
\end{lemma}
\begin{proof}
    If $f(z)$ is constant the statement holds, so assume $f(z)$ is non-constant and write
    \begin{equation*}
        f(z)=a_0+a_1 z+\ldots+a_d z^d,
    \end{equation*}
    where $a_d\neq 0$ and $d\geq 1$ is the degree of $f(z)$. Then we have
    \begin{equation*}
        0=f(z+1)-f(z)=d\, a_d\,  z^{d-1}+\mathcal{O}(z^{d-2}),
    \end{equation*}
 as $z\rightarrow \infty$,   where $\mathcal{O}(z^{d-2})$ is a polynomial of degree at most $d-2$ in $z$. Therefore, $d\equiv0$ mod $p$, correspondingly take $n\geq 1$ such that $d=n\,p$.

    Now consider the polynomial
    \begin{equation*}
        \widetilde{f}(z)=f(z)-a_d(z^p-z)^n\in \mathbb{F}_p[z].
    \end{equation*}
    This polynomial is also invariant under $z\mapsto z+1$, but its degree is strictly smaller than the degree of $f(z)$. By recursively applying the above argument a finite number of times, we will end up with a constant polynomial and obtain $f(z)$ expressed as a polynomial in $z^p-z$. The lemma follows.
\end{proof}

Note that as an immediately corollary of Lemma \ref{lem:invpol} we can conclude that
    \begin{equation}\label{eq:polyproduct}
        \prod_{k=0}^{p-1} (z+k) = z^p - z,
    \end{equation}
since the left-hand side is invariant under $z\mapsto z+1$, its highest order term is $z^p$ and $z=0$ is a root.
We now have all the ingredients necessary to prove Lemma \ref{lem:trace}.

\begin{proof}[Proof of Lemma \ref{lem:trace}]
Define the polynomial
\[ \chi_p(z) = \operatorname{Tr}\left[A(z+p-1)\cdot A(z+p-2)\cdot\ldots\cdot A(z+1)\cdot A(z)\right] .\]
As the trace is invariant under circular shifts we can immediately deduce that $\chi_p(z+1) = \chi_p(z)$. Furthermore, as the degree of $A(z)$ is at most two, we conclude that the degree of $\chi_p(z)$ is bounded by $2p$. Thus, application of Lemma \ref{lem:invpol} yields
\begin{equation*}
    \chi_p(z)=c_0+c_1 (z^p-z)+c_2(z^p-z)^2,
\end{equation*}
for some coefficients $c_i$, $0\leq i\leq 2$, which are polynomials in the entries of the $A_k$, $0\leq k\leq 2$, over $\mathbb{F}_p$. It remains to determine $c_1$ and $c_2$.

We introduce the notation $\mathbf{j}=(j_0,\ldots,j_{p-1})\in \{0,1,2\}^p$ for length $p$ tuples of indices and define the set of indices that sum up to $d$,
  \begin{equation*}
      J_d=\{\mathbf{j}\in \{0,1,2\}^p:j_0+\ldots+j_{p-1}=d\},
  \end{equation*}
for $0\leq d\leq 2p$. We then have a cyclic action on $J_d$ generated by the permutation
\begin{equation*}
    \sigma: J_d\rightarrow J_d, \qquad \mathbf{j}\mapsto (j_1,\ldots,j_{p-1},j_0).
\end{equation*}
Fixed points of $\sigma$ are given by $\mathbf{j}'s$ with all coordinates equal,  $j_0=j_1=...=j_{p-1}$. This is only possible if $p$ divides $d$. As $p$ is a prime, every other orbit under $\sigma$ has length $p$.

We  decompose $\chi_p(z)$ as follows,
    \begin{equation}\label{eq:trace_decomp}
       \chi_p(z)=\sum_{d=0}^{2p} H_d(z),
    \end{equation}
where
\begin{equation}\label{eq:H(z) sum}
   H_d(z)=\sum_{\mathbf{j}\in J_d}g_\mathbf{j}(z) \operatorname{Tr}\left[A_{j_{p-1}}A_{j_{p-2}}\cdot\ldots\cdot A_{j_1}A_{j_0}\right],
\end{equation}
with
\begin{equation*}
    g_\mathbf{j}(z)=(z+p-1)^{j_{p-1}}(z+p-2)^{j_{p-2}}\cdot\ldots\cdot (z+1)^{j_1}(z+0)^{j_0}.
\end{equation*}
Note that by construction, the degree of $H_d(z)$ is at most $d$. This means that to compute $c_1$ and $c_2$ we only need to consider contributions from $H_d(z)$ for $p\leq d \leq 2p$ in equation \eqref{eq:trace_decomp}. We work out these contributions separately for the cases $d=2p$, $d=2p-1$, $p<d<2p-1$ and $d=p$.\\
$\mathbf{d=2p}$: We have $J_{2p}=\{(2,2,\ldots,2,2)\}$ is a singleton and
    \begin{eqnarray}
        H_{2p}(z) &=& (z+p-1)^2...(z+1)^2z^2\operatorname{Tr}A_{2}^p,\nonumber\\
        &=& [(z+p-1)...(z+1)z]^2\operatorname{Tr}A_{2}^p,\nonumber\\
        &=& (z^p-z)^2\operatorname{Tr}A_{2}^p,\label{eq:c3}
    \end{eqnarray} 
    where the last equality is a direct application of equation \eqref{eq:polyproduct}. We conclude that $c_2 = \operatorname{Tr}A_{2}^p$.\\
$\mathbf{d=2p-1}$: In this case $J_{d}$ has $p$ elements, each with all indices equal to $2$ but one (whose value is $1$), so that
    \begin{align}
        H_{2p-1}(z) &= \left( \sum_{i=0}^{p-1}\frac{\prod_{k=0}^{p-1} (z+k)^2}{z+i} \right) \operatorname{Tr}A_{2}^{p-1}A_1
        = \left(\prod_{k=0}^{p-1} (z+k)\right) \left( \sum_{i=0}^{p-1}\frac{\prod_{k=0}^{p-1} (z+k)}{z+i} \right) \operatorname{Tr}A_{2}^{p-1}A_1\nonumber\\
        &= (z^p-z)  \frac{\partial}{\partial z} \left( \prod_{k=0}^{p-1} (z+k)\right) \operatorname{Tr}A_{2}^{p-1}A_1=  (z^p-z)  \frac{\partial}{\partial z}(z^p-z) \operatorname{Tr}A_{2}^{p-1}A_1\nonumber\\
        &=  -(z^p-z)\operatorname{Tr}A_{2}^{p-1}A_1. \label{eq:c2(i)}
    \end{align} 
    The first equality follows from the fact that the trace is invariant under circular shifts. The second to last equality follows from \eqref{eq:polyproduct}.
    \item $\mathbf{p< d <2p-1}$: 
As $p$ does not divide $d$, every orbit under $\sigma$ in $J_d$ has length $p$. Let us choose a set of representatives of the corresponding quotient,
\begin{equation*}
    J_d/\langle \sigma \rangle =\{[\mathbf{j}^{(r)}]:1\leq r \leq R\},
\end{equation*}
where $R=|J_d|/p$. We have, by the cycling invariance of traces of products of matrices,
\begin{equation*}
   H_d(z)=\sum_{r=1}^R\Big( \sum_{n=0}^{p-1}g_{\mathbf{j}^{(r)}}(z+n) \Big) \operatorname{Tr}\left[A_{j_{p-1}^{(r)}}A_{j_{p-2}^{(r)}}\cdot\ldots\cdot A_{j_1^{(r)}}A_{j_0^{(r)}}\right].
\end{equation*}
By Lemma \ref{lemma:z+n}, we know that each of the inner sums
\begin{equation*}
    \sum_{n=0}^{p-1} g_{\mathbf{j}^{(r)}}(z+n),
\end{equation*}
is a constant. In particular, they do not contribute to $c_1$ or $c_2$.\\
$\mathbf{d=p}$: Following the same arguments as those used for $p<d<2p-1$, we conclude that the only term that will contribute to $c_1$ is the fixed point $\mathbf{j} = (1,1,...,1)$. Thus, using \eqref{eq:polyproduct} we conclude that
\begin{equation}
    H_{p}(z) = (z^p-z)\operatorname{Tr}A_1^p+c, \label{eq:c2(ii)}
\end{equation}
for some constant $c$.

Combining equations \eqref{eq:c3}, \eqref{eq:c2(i)} and \eqref{eq:c2(ii)}, we obtain  equation \eqref{eq:trace with Ak} and thus the lemma.
\end{proof}

\section{Testing an integral on an algebraic solution}\label{sec:algebraic}
In this section, we check that $\mathcal{I}_7^\textrm{VI}$ is indeed an integral of motion for an algebraic solution of $\Psix$ over $\mathbb{F}_7$. We obtain the solution by reducing modulo $7$ the algebraic solution of $\Psix$ over $\mathbb{Q}$ named \textit{solution} $1$ in \cite{lisovyyalgebraic}, which is solution $20$ in \cite{boalchico}, given by
\begin{align*}
    f&=\frac{2(5s-2)(s^2+s+7)}{s(s+5)(4s^2-5s+10)},\\
    g&=\frac{s(s-4)(s+5)(4s^2-5s+10)}{15(s+2)^2(4s-7)(5s-2)},\\
    t&=\frac{27(5s-2)^2}{(s+5)(4s^2-5s+10)^2},
\end{align*}
with parameter values
\begin{equation*}
    \alpha_0=\frac{1}{3},\quad \alpha_1=-\frac{1}{3},\quad \alpha_2=\frac{1}{5},\quad \alpha_3=\frac{1}{5},\quad \alpha_4=\frac{2}{5}. 
\end{equation*}
It simplifies slightly when reduced modulo $7$, becoming
\begin{equation}\label{eq:algsolred7}
\begin{aligned}
    f&=\frac{6(s+1)^2}{(s+5)(s^2+4s+6)},\\
    g&=\frac{3(s+3)(s+5)(s^2+4s+6)}{(s+1)(s+2)^2},\\
    t&=\frac{5(s+1)^2}{(s+5)(s^2+4s+6)^2},
\end{aligned}
\end{equation}
with parameter values
\begin{equation}\label{eq:paramalg}
    \alpha_0=5,\quad \alpha_1=2,\quad \alpha_2=3,\quad \alpha_3=3,\quad \alpha_4=6. 
\end{equation}
Note that $f$ and $g$ are algebraic over $\mathbb{F}_7(t)$; we have
\begin{equation}\label{eq:algebraicimplicit}
f^5+5 f^4 t+2 f^3 t+3 f^2 t+5 f t^2+5 t^3=0,\quad f^2 g+6\, t f g+3 f+6 t=0.
\end{equation}

For the parameter values \eqref{eq:paramalg}, the integral of motion $\mathcal{I}_7^\textrm{VI}$ factorises as follows,
\begin{equation*}
    \mathcal{I}_7^\textrm{VI}=I_1^2I_2I_3I_4I_5,
\end{equation*}
where
\begin{align*}
    I_1=\,&f^2 g+6\, t f g +3 f+6 t,\\
    I_2=\,&f^2 g+6 f g+5 f+6,\\
    I_3=\,&6 t f^3 g^3+2 t f^2 g^3+5 t f^2 g^2+6 t f g^3+3 t f g^2+2 t f g+6 t g^2+3 t g+4 t+f^4 g^3+5 f^3 g^3+4 f^3 g^2+f^2 g^3\\
    &+4 f^2 g+3 f g^2+2 f g+3 f+3 g+1,\\
    I_4=\,&4 + 3 t + 3 t^2 + 6 t f + 5 f^2 + 3 t g + 4 t^2 g + f g + 5 t^2 f g + 
 2 f^2 g + 6 f^3 g + t^2 g^2 + t f g^2 + 4 t^2 f g^2 + 5 f^2 g^2\\
 &+ 
 2 t^2 f^2 g^2 + 3 f^3 g^2 + 6 t f^3 g^2 + 6 f^4 g^2 + t^2 f g^3 + 
 5 t f^2 g^3 + 5 t^2 f^2 g^3 + f^3 g^3 + 4 t f^3 g^3 + t^2 f^3 g^3 + 
 5 f^4 g^3\\
 &+ 5 t f^4 g^3 + f^5 g^3,\\
 I_5=\,&t^2 f^4 g^5+5 t^2 f^3 g^5+3 t^2 f^3 g^4+t^2 f^2 g^5+2 t^2 f^2 g^4+5 t^2 f^2 g^3+2 t^2 f g^4+3 t^2 f g^3+5 t^2 f g^2+t^2 g^3+6 t^2 g^2\\
 &+5 t f^5 g^5+4 t f^4 g^5+4 t f^4 g^4+5 t f^3 g^5+4 t f^3 g^4+6 t f^2 g^4+2 t f^2 g^3+2 t f^2 g^2+t f g^3+5 t f g^2+3 t f g+5 t g\\
 &+f^6 g^5+5 f^5 g^5+f^4 g^5+f^4 g^4+6 f^4 g^3+6 f^3 g^4+f^3 g^3+3 f^3 g^2+2 f^2 g^3+4 f^2 g+4 f g+5 f+5 g+4.
\end{align*}
It follows from the second algebraic equation in \eqref{eq:algebraicimplicit} that $I_1=0$. So we see that $\mathcal{I}_7^\textrm{VI}$ is indeed constant, and turns out to equal zero, for this particular solution of $\Psix$ over $\mathbb{F}_7$.


\begin{bibdiv}
 \begin{biblist}


\bib{adler94}{article}{
   author={Adler, V. \`E.},
   title={Nonlinear chains and Painlev\'e{} equations},
   journal={Phys. D},
   volume={73},
   date={1994},
   number={4},
   pages={335--351}
}



\bib{arinkinbor2006}{article}{
   author={Arinkin, D.},
   author={Borodin, A.},
   title={Moduli spaces of $d$-connections and difference Painlev\'e{}
   equations},
   journal={Duke Math. J.},
   volume={134},
   date={2006},
   number={3},
   pages={515--556}
}



\bib{boalchico}{article}{
   author={Boalch, P.},
   title={The fifty-two icosahedral solutions to Painlev\'e{} VI},
   journal={J. Reine Angew. Math.},
   volume={596},
   date={2006},
   pages={183--214},
   issn={0075-4102}
}


\bib{borodin2006}{article}{
   author={Arinkin, D.},
   author={Borodin, A.},
   title={Moduli spaces of $d$-connections and difference Painlev\'e{}
   equations},
   journal={Duke Math. J.},
   volume={134},
   date={2006},
   number={3},
   pages={515--556}
}



\bib{buium_manin_15}{article}{
   author={Buium, A.},
   author={Manin, Yuri I.},
   title={Arithmetic differential equations of Painlev\'e{} VI type},
   conference={
      title={Arithmetic and geometry},
   },
   book={
      series={London Math. Soc. Lecture Note Ser.},
      volume={420},
      publisher={Cambridge Univ. Press, Cambridge},
   },
   isbn={978-1-107-46254-0},
   date={2015},
   pages={114--138},
   review={\MR{3467121}},
}





\bib{NIST:DLMF}{misc}{
       title = {NIST Digital Library of Mathematical Functions},
howpublished = {\url{https://dlmf.nist.gov/}, Release 1.2.6 of 2026-03-15},
         url = {https://dlmf.nist.gov/},
        note = {F.~W.~J. Olver, A.~B. {Olde Daalhuis}, D.~W. Lozier, B.~I. Schneider,
                R.~F. Boisvert, C.~W. Clark, B.~R. Miller, B.~V. Saunders,
                H.~S. Cohl, and M.~A. McClain, eds.}
}




\bib{joshiroffelsen}{article}{
  author = {Joshi, N.},
  author = {Roffelsen, P.},
  title = {On integrals of non-autonomous dynamical systems in finite characteristic},
  journal = {arXiv preprint},
  year = {2026},
  eprint = {2602.09298},
  archivePrefix = {arXiv},
  primaryClass = {nlin.SI},
  doi = {10.48550/arXiv.2602.09298}
}



\bib{KNY2017}{article}{
  author  = {Kajiwara, K.},
  author  = {Noumi, M.},
  author  = {Yamada, Y.},
  title   = {Geometric Aspects of Painlev{\'e} Equations},
  journal = {J. Phys. A},
  volume  = {50},
  number  = {7},
  year    = {2017}
}


\bib{kankisigma}{article}{
   author={Kanki, M.},
   title={Integrability of discrete equations modulo a prime},
   journal={SIGMA Symmetry Integrability Geom. Methods Appl.},
   volume={9},
   date={2013},
   pages={Paper 056, 8},
   review={\MR{3141527}}
}


\bib{kankirims}{article}{
   author={Kanki, M.},
   author={Mada, J.},
   author={Tokihiro, T.},
   title={Discrete Painlev\'e{} equations and discrete KdV equation over
   finite fields},
   conference={
      title={The breadth and depth of nonlinear discrete integrable systems},
   },
   book={
      series={RIMS K\^oky\^uroku Bessatsu},
      volume={B41},
      publisher={Res. Inst. Math. Sci. (RIMS), Kyoto},
   },
   date={2013},
   pages={125--145},
}


\bib{lisovyyalgebraic}{article}{
   author={Lisovyy, O.},
   author={Tykhyy, Y.},
   title={Algebraic solutions of the sixth Painlev\'e{} equation},
   journal={J. Geom. Phys.},
   volume={85},
   date={2014},
   pages={124--163}
}

\bib{macmillansondow}{article}{
   author={MacMillan, K.},
   author={Sondow, J.},
   title={Proofs of power sum and binomial coefficient congruences via
   Pascal's identity},
   journal={Amer. Math. Monthly},
   volume={118},
   date={2011},
   number={6},
   pages={549--551}
}

\bib{nagloo_21}{article}{
   author={Nagloo, J.},
   title={Model theory and differential equations},
   journal={Notices Amer. Math. Soc.},
   volume={68},
   date={2021},
   number={2},
   pages={177--185}
}

\bib{nagloo_pillay_14}{article}{
   author={Nagloo, J.},
   author={Pillay, A.},
   title={On the algebraic independence of generic Painlev\'e{}
   transcendents},
   journal={Compos. Math.},
   volume={150},
   date={2014},
   number={4},
   pages={668--678}
}






\bib{nakazono_lax}{article}{
   author={Nakazono, N.},
   title={Reduction of lattice equations to the Painlev\'e{} equations: $\rm
   P_{IV}$ and $\rm P_V$},
   journal={J. Math. Phys.},
   volume={59},
   date={2018},
   number={2},
}





\bib{ormerodrains2017}{article}{
   author={Ormerod, C. M.},
   author={Rains, E.},
   title={A symmetric difference-differential Lax pair for Painlev\'e{} VI},
   journal={J. Integrable Syst.},
   volume={2},
   date={2017},
   number={1},
   pages={xyx003, 20}
}


\bib{joshiroffelsencrystal}{article}{
   author={Joshi, N.},
   author={Roffelsen, P.},
   title={On the crystal limit of the $q$-difference sixth Painlev\'e{}
   equation},
   journal={J. Nonlinear Sci.},
   volume={35},
   date={2025},
   number={1},
   pages={Paper No. 31, 30}
}

\bib{joshiroffelsen1}{article}{
  author = {Joshi, N.},
  author = {Roffelsen, P.},
  title = {Arithmetic dynamics of a discrete Painlev\'e equation},
  journal = {Journal of Physics A: Mathematical and Theoretical},
  year = {2026},
volume = {59},
number = {19}
}




\bib{umemura1998painleve}{article}
{
	author = {Umemura, H.},
	journal = {Sugaku Expositions},
	number = {1},
	pages = {77--100},
	publisher = {Providence, RI, USA: The Society, c1988-},
	title = {{P}ainlev\'e equations and classical functions},
	volume = {11},
	year = {1998}}



\end{biblist}
\end{bibdiv}
\end{document}